\documentclass[11pt]{article}
\usepackage[margin=1in]{geometry}
\usepackage[english]{babel}
\usepackage{xr}
\usepackage[inline]{enumitem}
\usepackage{setspace}
\usepackage[normalem]{ulem}

\usepackage{tikz}
\usetikzlibrary{arrows}

\usepackage{mathrsfs,hyperref}
\usepackage[]{amsmath}
\usepackage{amsthm}
\usepackage{fix-cm}
\usepackage[]{amssymb}
\usepackage[]{latexsym}
\usepackage[right]{eurosym}
\usepackage[T1]{fontenc}
\usepackage[]{graphicx}
\usepackage[]{epsfig}
\usepackage{fancyhdr}
\usepackage{bbm}
\usepackage{pstricks}
\usepackage{multirow}
\usepackage[numbers]{natbib}
\usepackage{subcaption}
\makeatletter

\newcommand{\bbr}{\mathbb{R}}
\newcommand{\E}{\mathbb{E}}

\newcommand{\bbn}{\mathbb{N}}
\renewcommand{\P}{\mathbb{P}}

\newcommand{\bba}{\mathbb{A}}

\newcommand{\bbt}{\mathbb{T}}

\newcommand{\fcal}{\mathcal{F}}

\newcommand{\dcal}{D}

\newcommand{\ncal}{\mathcal{N}}
\newcommand{\lcal}{\mathcal{L}}

\newcommand{\Gd}{\Delta}

\newcommand{\Gl}{\Lambda}

\newcommand{\Dt}{{\Delta t}}

\newcounter{modcount}
\newcommand{\modulo}[2]{%
\setcounter{modcount}{#1}\relax
\ifnum\value{modcount}<#2\relax
\else\relax
\addtocounter{modcount}{-#2}\relax
\modulo{\value{modcount}}{#2}\relax
\fi}
\newcommand{\tablepictures}[4][c]{\begin{tabular}[#1]{@{}c@{}}#2\vspace{0.5cm}\\(\alph{#4}) #3\end{tabular}}
\newcounter{gridsearch}
\newcommand{\tabpic}[2]{
    \stepcounter{gridsearch}
    \modulo{\thegridsearch}{2}
    \ifnum\value{modcount}=0
        \tablepictures[t]{#1}{#2}{gridsearch}\\[2.0cm]
    \else
        \tablepictures[t]{#1}{#2}{gridsearch}&~&
    \fi
}

\makeatother
\newtheorem{lemma}{Lemma}[section]
\newtheorem{proposition}[lemma]{Proposition}
\newtheorem{theorem}[lemma]{Theorem}
\newtheorem{corollary}[lemma]{Corollary}
\newtheorem{definition}[lemma]{Definition}
\newtheorem{example1}[lemma]{Example}
\newtheorem{rem1}[lemma]{Remark}
\newtheorem{assumption}[lemma]{Assumption}
\newtheorem{alg1}[lemma]{Algorithm}
\newtheorem{me1}[lemma]{Mechanism}

\newenvironment{remark}{\begin{rem1}\rm}{\end{rem1}}
\newenvironment{example}{\begin{example1}\rm}{\end{example1}}

\newenvironment{alg}{\begin{alg1}\rm}{\end{alg1}}

\usepackage{color}

\newcommand{\T}{\top}
\newcommand{\diag}{\operatorname{diag}}

\newcommand\ind[1]{\mathbf{1}_{\{#1\}}}

\begin{document}

\title{Dynamic clearing and contagion in financial networks}
\author{Tathagata Banerjee\thanks{Washington University in St.\ Louis, Department of Electrical and Systems Engineering, St.\ Louis, MO 63130, USA.} \and
Alex Bernstein\thanks{University of California Santa Barbara, Department of Statistics \& Applied Probability, Santa Barbara, CA 93106.} \and
Zachary Feinstein\thanks{Stevens Institute of Technology, School of Business, Hoboken, NJ 07030, USA, {\tt zfeinste@stevens.edu}. 
The author acknowledges the support from NSF IUCRC CRAFT center research grant (2113906) for this research. The opinions expressed in this publication do not necessarily represent the views of NSF IUCRC CRAFT. The author thanks Park Avenue Finance for their subject matter expertise and insight into this topic. The opinions expressed in this publication do not necessarily represent the view of Park Avenue Finance.}~\thanks{Corresponding author}}
\date{\today}
\maketitle
\abstract{
In this paper we introduce a generalized extension of the Eisenberg-Noe model of financial contagion to allow for time dynamics of the interbank liabilities, including a dynamic examination of default risk. This framework separates the cash account and long-term capital account to more accurately model the health of a financial institution.  In doing so, such a system allows us to distinguish between delinquency and default as well as between defaults resulting from either insolvency or illiquidity. 
}

\section{Introduction}\label{sec:intro}

Financial networks and the contagion of bank failures have been widely studied beginning with the seminal work on financial payment networks by Eisenberg \& Noe \cite{EN01}.  The 2007-2009 financial crisis and credit crunch showed the severe impacts that systemic crises can have on the financial sector and the economy as a whole.  As the cost of such cascading events is tremendous, the modeling of such events is imperative.
Recently there have been significant studies on modeling financial systemic risk and financial contagion.  A major class of models for systemic risk is those based on network models from \cite{EN01} and, its extension which introduced dead-weight losses from defaults, \cite{RV13}.  Notably, these network models generally approach the problem only in a static setting.  

In this work, we will focus on a time-dynamic interbank network model.  Doing so introduces two primary complications beyond prior static systems.  First, system dynamics between default times need to be introduced and studied.  This was considered in \cite{sonin2020continuous} under the special -- time-homogeneous -- case; within this work we consider heterogeneous network structures over time.  Second, the introduction of time dynamics ultimately decouples the cash account and capital account of the firms which coincide exactly in the static time setting.  Importantly, this distinction between liquidity and capital leads to the possibility of delinquency without default as well as differing default notions, i.e., illiquidity and insolvency, which need not coincide. The impact of delinquency, before default, on the health of the financial system was made readily apparent in Spring 2023 with the failures of, e.g., Silicon Valley Bank and First Republic Bank, with the liquidity panic occurring over the course of months.

The organization of this paper is as follows.  In the remainder of this introduction, we provide a detailed literature review and highlight the primary contributions and takeaways of this work.
In Section~\ref{sec:fednow} we highlight stylized facts about real-time gross settlement systems, as implemented in practice, that our proposed model will focus on.
In Section~\ref{sec:setting} we will provide a review of the static Eisenberg-Noe framework.  Of particular interest, in this section, we consider the clearing to be in terms of the equity and losses of the firms, as considered in, e.g., \cite{veraart2017distress,barucca2016valuation} rather than payments as originally studied in \cite{EN01}.
In Section~\ref{sec:continuous} we propose a continuous-time formulation for the Eisenberg-Noe model between default times.  We consider existence and uniqueness of the clearing solutions under external assets modeled by stochastic processes.  
In Section~\ref{sec:default} we present an extended model of the Eisenberg-Noe framework which includes the impact of early defaults.  These defaults can be triggered either by insolvency (negative or insufficient capital) or illiquidity (negative or insufficient cash) events in the interbank network.  With this model, we study the implications of asset dynamics on the type and timing of defaults.  Numerical case studies conclude in Section~\ref{sec:numerics}.
The proofs of the main technical results are provided in Appendix~\ref{sec:continuous-proof}.  Additionally, a brief consideration for a discrete-time formulation, and its use in deriving the continuous-time formulation, is provided in Appendix~\ref{sec:discrete-top}.  
Appendix~\ref{sec:discussion-EN} provides a formulation of the static Eisenberg-Noe system as the solution of our inter-default time dynamic model (Section~\ref{sec:continuous}) under a specific dynamic network construction; this independently replicates the results of~\cite{sonin2020continuous}.

\subsection{Literature Review}

Interbank networks were studied first in \cite{EN01} to model the spread of defaults in the financial system.  In the Eisenberg-Noe framework, financial firms must satisfy their liabilities by transferring assets.  One firm being unable to meet its liabilities due to a shortfall of assets can cause other firms to default on some of their liabilities as well, causing a cascading failure in the financial system.  
The Eisenberg-Noe framework has been extended in multiple directions. \cite{RV13} included bankruptcy costs to the interbank network which have been extended in numerous subsequent works (see, e.g.,~\cite{EGJ14,GY14,AW_15,CCY16,veraart2017distress}).  This has also been studied empirically, e.g. in \cite{Elsinger2006_Risk}.

In this work, we will focus on adding the time dynamics to the interbank network approach.  In fact, the conclusion of \cite{EN01} provides a discussion of future extensions, one of which is the inclusion of multiple clearing dates.  This has been studied directly in \cite{CC15,ferrara2019systemic,bardoscia2017pathways}.  Additionally, \cite{KV16} considers a similar approach to model financial networks with multiple maturities.  All of these works only consider clearing at discrete times, which does not, e.g., permit defaults between clearing times due to a shock to a firm's capital.  \cite{sonin2020continuous} presents a continuous-time clearing model that exactly replicates the static Eisenberg-Noe framework without considerations for early defaults. \cite{fs2019,feinstein2020dynamic,feinstein2021contagious} extends this work in the network model of~\cite{GK10} by presenting a special case in which the large bank limit can be formulated as a McKean-Vlasov equation.  
\cite{lipton2016modern} presents a related framework for dynamic financial contagion with early defaults. 
\cite{chen2021financial} introduced a semi-dynamic crisis index and conducted an empirical study for systemic risk. 
In a related setting, \cite{calafiore2022,calafiore2023} considered a discrete-time Eisenberg-Noe system from a controls perspective.

As we will elaborate on further below, and in particular in Section~\ref{sec:default}, within this time-dynamic interbank network we differentiate types of defaults due to illiquidity and insolvency. Illiquidity has been studied previously with a primary focus on fire sales and illiquid assets.  We refer the interested reader to, e.g.,~\cite{CFS05,CLY14,GLT15,AFM16,feinstein2015illiquid,CS19}. Herein, we differ from those works in that insolvency can fundamentally differ from illiquidity rather than the prior assumption that illiquidity precedes insolvency; that is, we allow for a firm to be (i) solvent and liquid; (ii) solvent and illiquid; (iii) insolvent and liquid; or (iv) insolvent and illiquid which, as far as the authors are aware, is novel in the systemic risk literature.

\subsection{Primary Contributions and Innovations}

Within this work, we construct a time-dynamic Eisenberg-Noe system of interbank obligations that captures the features of said system in continuous-time stochastic setting.  This constructed system is general enough to allow for dynamic and heterogeneous obligations; such generality encodes a realistic and complicated seniority structure that cannot be captured by static models or previous work in continuous-time Eisenberg-Noe systems (e.g. \cite{sonin2020continuous,chen2021financial}).
Furthermore, we encode a grace period during which time a stressed institution can recover from delinquency before a default is actualized. This distinction between delinquency and default follow as in bankruptcy laws in all major economies; we note that prior works on dynamic Eisenberg-Noe models (e.g., \cite{feinstein2020dynamic}) specify default rules that only consider the current state of the system without consideration for historical events as taken herein. 
As will be demonstrated in simple numerical case studies, this construction introduces complex nonlinearities in the system behavior. For instance, we find that tighter regulatory requirements (i.e., shorter grace periods in which a bank can recover from delinquency) precipitate increased defaults, but the wealth of the financial system is non-monotonic as the increased defaults are counteracted with higher recovery rates on payments.  
Finally, in relation to default times, the time dynamics also allow us to distinguish between defaults due to illiquidity (i.e., due to a liquidity shortfall which causes a bank to fail to make its obliged payments on time) and insolvency (e.g., due to either capital regulations or shareholder choice to restructure to maximize long-term wealth); in the static setting of \cite{EN01}, these default notions coincide exactly. However, we note that the regulatory response to prevent these two types of defaults differ and, therefore, being able to distinguish between these events is of great importance for systemic risk management. 

\section{Real-Time Gross Settlement}\label{sec:fednow}
Before continuing to the main results of this work, we wish to review the practical applications of a continuous-time dynamic clearing system. Though most prior research has focused on discrete settlement times, there are numerous examples of real-time gross settlement [RTGS] systems used around the world. An RTGS system, as its name indicates, allows for obligations to be paid continuously whenever they come due without any netting of obligations before payment.

We, specifically, wish to highlight the newly created FedNow payment system that was introduced in Summer 2023. This RTGS system operates around-the-clock all year long with payments made in real-time. As of this writing, FedNow functionally operates in a failure-to-pay model without the typical overdraft protection of Fedwire as the net debit cap of (almost) all institutions is set to 0. This failure-to-pay system means that any liquidity shortfall on the obligor's part will result in realized counterparty risk for the obligee. However, rather than write these missed payments off as a loss, the FedNow system is structured so that these payments are made once sufficient liquidity of the obligor is made available.\footnote{This stylized RTGS system is comparable to wholesale central bank digital currency systems which are under active study in central banks globally.}

Though we allow unpaid liabilities to accrue over time, we also allow the banks to default on their obligations. Within the ISDA Master Agreement, there are 8 forms of default events. We consolidate these rules into 2 types:
\begin{itemize}
\item Defaults due to \emph{illiquidity} in which the firm fails to cover its overdue liabilties within a pre-specified grace period. Within the ISDA Master Agreement this can fall under either Failure to Pay (grace period until the end of the business day) or a Breach of Agreement (with a 30 day grace period) depending on the nature of the obligation. Any firm that has a liquidity shortfall, but not yet in default, is called \emph{delinquent}.
\item Defaults due to \emph{insolvency} in which the firm declares bankruptcy. Such an event would occur if the firm has negative equity on its balance sheet (accounting for all future assets and liabilities).
\end{itemize}
Collectively, we define any bank that is either in delinquency or in default as being \emph{distressed} so as to encompass both terms.

Within this work, inspired by FedNow, we will model a failure-to-pay RTGS system in which unmet obligations automatically roll-forward until either the obligation is covered in full or the obligor defaults (see Section~\ref{sec:continuous}). Furthermore, we model defaults following the rules set out within the ISDA Master Agreement based on insolvency or illiquidity (see Section~\ref{sec:default}).

\section{Eisenberg-Noe Static Clearing}\label{sec:setting}
We begin with some simple notation that will be consistent for the entirety of this paper.  Let $x,y \in \bbr^n$ for some positive integer $n$, then
\[x \wedge y = \left(\min(x_1,y_1),\min(x_2,y_2),\ldots,\min(x_n,y_n)\right)^\T,\] $x^- = -(x \wedge 0)$, and $x^+ = (-x)^-$.  Further, to ease notation, we will denote $[x,y] := [x_1,y_1] \times [x_2,y_2] \times \ldots \times [x_n,y_n] \subseteq \bbr^n$ to be the $n$-dimensional compact interval for $y - x \in \bbr^n_+$.  Similarly, we will consider $x \leq y$ if and only if $y - x \in \bbr^n_+$.

Throughout this paper we will consider a network of $n$ financial institutions.  We will denote the set of all banks in the network by $\ncal := \{1,2,\ldots,n\}$.  Often we will consider an additional node $0$, which encompasses the entirety of the financial system outside of the $n$ banks; this node $0$ will also be referred to as society or the societal node.  The full set of institutions, including the societal node, is denoted by $\ncal_0 := \ncal \cup \{0\}$. We refer to \cite{feinstein2014measures,GY14} for further discussion of the meaning and concepts behind the societal node.

We will be extending the model from \cite{EN01} in this paper.  In that work, any bank $i \in \ncal$ may have obligations $L_{ij} \geq 0$ to any other firm or society $j \in \ncal_0$.  We will assume that no firm has any obligations to itself, i.e., $L_{ii} = 0$ for all firms $i \in \ncal$, and the society node has no liabilities at all, i.e., $L_{0j} = 0$ for all firms $j \in \ncal_0$.  Thus the \emph{total liabilities} for bank $i \in \ncal$ is given by $\bar p_i := \sum_{j \in \ncal_0} L_{ij} \geq 0$ and relative liabilities $\pi_{ij} := \frac{L_{ij}}{\bar p_i}$ if $\bar p_i > 0$ and arbitrary otherwise; for simplicity, in the case that $\bar p_i = 0$, we will let $\pi_{ij} = \frac{1}{n}$ for all $j \in \ncal_0 \backslash \{i\}$ and $\pi_{ii} = 0$ to retain the property that $\sum_{j \in \ncal_0} \pi_{ij} = 1$.
On the other side of the balance sheet, all firms are assumed to begin with some amount of \emph{external assets} $x_i \geq 0$ for all firms $i \in \ncal_0$.
The resultant \emph{clearing payments}, under a pro-rata repayment (i.e., no priority of payments) assumption, satisfy the fixed point problem in payments $p \in [0,\bar p]$
\begin{equation}\label{eq:EN-p}
p = \bar p \wedge \left(x + \Pi^\T p\right).
\end{equation}
That is, each bank pays the minimum of what it owes ($\bar p_i$) and what it has ($x_i + \sum_{j \in \ncal} \pi_{ji} p_j$).
The resultant vector of \emph{cash accounts} for all firms is given by
\begin{equation}\label{eq:equity}
V = x + \Pi^\T p - \bar p.
\end{equation}  
Noting that payments can be written as a simple function of the cash accounts ($p = \bar p - V^-$, recalling the aforementioned notation that $V^- = -(V\wedge 0)$), we provide the following proposition.  We refer also to \cite{veraart2017distress,barucca2016valuation,banerjee2017insurance} for similar notions of utilizing clearing cash accounts instead of clearing payments.  We refer to Appendix~\ref{sec:FDA} for the fictitious default algorithm that can be used to compute these clearing cash accounts.
\begin{proposition}\label{prop:EN-e}
A vector $p \in [0,\bar p]$ is a clearing payments in the Eisenberg-Noe setting \eqref{eq:EN-p} if and only if $p = [\bar p - V^-]^+$ for some $V \in \bbr^{n+1}$ satisfying the following fixed point problem
\begin{equation}\label{eq:EN-e}
V = x + \Pi^\T [\bar p - V^-]^+ - \bar p.
\end{equation}
Vice versa, a vector $V \in \bbr^{n+1}$ is a clearing cash account (i.e., satisfying \eqref{eq:EN-e}) if and only if $V$ is defined as in \eqref{eq:equity} for some clearing payments $p \in [0,\bar p]$ as defined in the fixed point problem \eqref{eq:EN-p}.
\end{proposition}
\begin{proof}
We will prove the first equivalence only, the second follows similarly.

Let $p \in [0,\bar p]$ be a clearing payment vector. Define the cash account vector $V$ by \eqref{eq:equity}, then it is clear that $V^- = \bar p - p$ by definition as well, i.e., $p = \bar p - V^- \geq 0$.  Thus from \eqref{eq:equity} we immediately recover that the cash account vector $V$ must satisfy \eqref{eq:EN-e}.

Let $p = [\bar p - V^-]^+$ for some cash account vector $V \in \bbr^{n+1}$ satisfying \eqref{eq:EN-e}.  By construction we find
\begin{align*}
p &= [\bar p - V^-]^+ = \bar p - \left(x + \Pi^\T [\bar p - V^-]^+ - \bar p\right)^-
= \bar p - \left(x + \Pi^\T p - \bar p\right)^- = \bar p \wedge \left(x + \Pi^\T p\right).
\end{align*}
We note that $\bar p \geq \left(x + \Pi^\T [\bar p - V^-]^+ - \bar p\right)^-$ can be shown trivially.
\end{proof}

Due to the equivalence of the clearing payments and clearing cash accounts provided in Proposition~\ref{prop:EN-e}, we are able to consider the Eisenberg-Noe system as a fixed point of equity and losses rather than payments.
In \cite{EN01,csoka_herings} results for the existence and uniqueness of the clearing payments (and thus for the clearing cash accounts as well) are provided.  In fact, it can be shown that there exists a unique clearing solution in the Eisenberg-Noe framework so long as $L_{i0} > 0$ for all firms $i \in \ncal$.  We will take advantage of this result later in this paper.  This is a reasonable assumption (as discussed in, e.g., \cite{GY14}) as obligations to society include, e.g., deposits to the banks.

Before continuing to our dynamic model, we wish to consider a reformulation of~\eqref{eq:EN-e} with a functional \emph{relative exposures matrix} $A$, i.e., 
\begin{equation}\label{eq:static-A}
a_{ij}(V) := \begin{cases} \pi_{ij} &\text{if } \bar p_i \geq V_i^- \\ \frac{L_{ij}}{V_i^-} &\text{if } \bar p_i < V_i^- \end{cases} \quad \forall i,j \in \ncal_0,
\end{equation}
so that the clearing cash account problem~\eqref{eq:EN-e} need not explicitly depend on the positive part $[\bar p - V^-]^+$.
For mathematical convenience, here we introduce the functional matrix $A: \bbr^{n+1} \to [0,1]^{(n+1)\times(n+1)}$ to be the relative exposure matrix.  We let $a_{ij}(V)V_i^-$ describe the negative impact firm $i$'s losses have on firm $j$'s cash account.  The truncation at $0$ endogenously encodes the notion of limited liability so that the losses of bank $i$ to all of its counterparties are capped by $L_{i\cdot}$.  
That is,
\[L^\T\vec{1} - A(V)^\T V^- = \Pi^\T [\bar p - V^-]^+\]
for any $V \in \bbr^{n+1}$.  This formulation is such that if the positive part were removed from the right hand side, the relative exposures $A$ would be defined exactly as the relative liabilities $\Pi$ by construction.  In particular, we will define the relative exposures element-wise and pointwise so as to encompass the limited exposures as in~\eqref{eq:static-A}.  If $\bar p_i > 0$ then we can simplify this further as $a_{ij}(V) = L_{ij} / \max\{\bar p_i , V_{i}^-\}$. 
The relative exposures matrix will be used in Section~\ref{sec:continuous} below to aid in the construction of a dynamical system representation of the clearing problem. 

Using the notation and terms above, we can rewrite~\eqref{eq:EN-e} as a (piecewise) linear system as:
\begin{align}
\nonumber V &= x + \Pi^\T[\bar p - V^-]^+ - \bar p = x + L^\T\vec{1} - A(V)^\T V^- - L\vec{1} \\
\label{eq:EN-e-inv} &= [I - A(V)^\T \Lambda(V)]^{-1} (x - [I - A(V)^\T] L\vec{1})
\end{align}
where $\Lambda = \Lambda(V) := \diag(\ind{V < 0})$, which we herein refer to as a \emph{delinquency matrix} is the diagonal matrix of defaulting banks.  The matrix $[I - A(V)^\T \Lambda(V)]^{-1}$ is often called the network multiplier (see, e.g.,~\cite{CLY14,feinstein2017sensitivity}) and captures the propagation of losses throughout the financial network. 
The following proposition proves the invertibility for the network multiplier provided every bank has some obligation to society. We wish to note that, because the delinquency matrix $\Lambda(V)$ is a function of the cash account $V$, the clearing cash account~\eqref{eq:EN-e-inv} does not simply follow from a single matrix inverse., but must be approached as a fixed point problem.

\begin{proposition}\label{prop:Leontief}
For any relative exposure matrix $A \in \bba := \{A \in [0,1]^{(n+1) \times (n+1)} \; | \; A\vec{1} \leq \vec{1}, \; a_{ii} = 0, \; a_{i0} > 0 \; \forall i \in \ncal_0\}$
and any delinquency matrix $\Lambda \in \{0,1\}^{(n+1) \times (n+1)}$ such that $\Lambda_{00} = 0$ and $\Lambda_{ij} = 0$ for $i \neq j$, the matrix $I - A^\T \Lambda$ is invertible with Leontief form, i.e., $(I-A^\T \Lambda)^{-1}=\sum_{k=0}^\infty(A^\T \Lambda)^k$.
\end{proposition}
\begin{proof}
By inspection, for any $A \in \bba$, $(I-A^\T \Lambda)(I+A^\T(I-\Lambda A^\T)^{-1} \Lambda)=I$, i.e., the form of the inverse is provided by $I + A^\T (I - \Lambda A^\T)^{-1} \Lambda$.  We refer to \cite[Theorem 2.6]{feinstein2017sensitivity} for a detailed proof that $(I - \Lambda A^\T)^{-1}$ is nonsingular and is provided by the Leontief inverse.  Therefore, by construction
\begin{align*}
(I-A^\T \Lambda)^{-1}&=(I+A^\T(I-\Lambda A^\T)^{-1} \Lambda)
= I+ A^\T \left(\sum_{k=0}^\infty (\Lambda A^\T)^k\right) \Lambda\\
&= I + \sum_{k=0}^\infty A^\T \Lambda (A^\T)^k \Lambda^k
= I + \sum_{k=0}^\infty (A^\T \Lambda)^{k+1}
=\sum_{k=0}^\infty(A^\T \Lambda)^k.
\end{align*}
\end{proof}

Before continuing to our discussion of time-dynamics, we wish to introduce a simple network which we will use as running example for the rest of this work.
\begin{example}\label{ex:2bank-static}
Consider a $n = 2$ bank system. Let bank $1$ have obligations $L_{12} = L_{10} = 2$ while bank $2$ has obligations $L_{21} = 3$ and $L_{20} = 1$. Finally, let the external assets for these two banks be given by $x_1 = x_2 = 2.1$. It is easy to verify that the \emph{unique} clearing cash accounts are given by:
\[V = (3 \, , \, 1.1 \, , \, 0.1)^\T.\]
Notably, there are no defaults within this system.
\end{example}

\begin{remark}\label{rem:K=V_static}
We conclude this review of the static network setting to consider the interpretation of the cash account $V$. In this setting in which all liabilities mature at the current time, the cash account also provides the capital $K$ of the bank. As we will see in the subsequent sections, when we introduce dynamics to any elements of the network, these two financial constructs no longer need to coincide.
\end{remark}

\section{Inter-Default Time Clearing in a Continuous-Time System}\label{sec:continuous}
Consider now a continuous set of clearing times $\bbt = [0,T]$ for some (finite) terminal time $T < \infty$.
We will use the notation from \cite{cont2013ito} such that the process $Z: \bbt \to \bbr^n$ has value of $Z(t)$ at time $t \in \bbt$ and history $Z_t := (Z(s))_{s \in [0,t]}$.  We will now construct an extension of the continuous-time setting of \cite{sonin2020continuous} in that we allow for liabilities to change over time and for firms to have stochastic external assets.

\begin{remark}
Within this section, we assume that no defaults occur within $\bbt$, i.e., we do not consider the dynamics that occur during such events. We do this so as to highlight the inter-default time dynamics. 
We wish to note that the results within this section can be adjusted so that $\bbt = [\tau_0,\tau_1]$ for any two stopping times $0 \leq \tau_0 < \tau_1 \leq T$ (for some $T < \infty$) without altering any subsequent results; these stopping times can be viewed as the time of default events which will be explored in detail within Section~\ref{sec:default}.
\end{remark}

\subsection{Setting and Dynamical System}
In order to construct a continuous-time model we will begin by considering our network parameters of cash flows and nominal liabilities.
Generally, we will consider the external assets $x: \bbt \to \bbr^{n+1}$ and nominal liabilities $L: \bbt \to \bbr^{(n+1) \times (n+1)}_+$ to be functions of the clearing time, i.e., as assets and liabilities with different maturities.  For simplicity, throughout we are considering the discounted assets and liabilities. We take the discounted values so as to simplify notation without explicitly considering the risk-free rate.
We wish to note that the total liabilities are given by $L(t)\vec{1}$ and the total incoming interbank obligations are given by $L(t)^\T\vec{1}$.
Throughout this formulation we will consider the term $dx(t)$ of marginal change in external assets at time $t$, i.e., the external cash flows.  Similarly, we will consider $dL(t)$ to be the marginal change in nominal liabilities matrix at time $t$ as a function of time; we note that by assumption $dL_{ij}(t) \geq 0$ for all firms $i,j \in \ncal_0$ as, without any payments made, total liabilities should accumulate over time. We assume, as in the static setting in Section~\ref{sec:setting} that all obligations take the same seniority.\footnote{Extensions to account for priority of obligations as in, e.g.,~\cite{E07} can be accomplished but go beyond the scope of this work.}  Theorem~\ref{thm:continuous} provides existence and uniqueness of the clearing cash accounts driven by $(dx,dL)$ when $x(t) = \int_0^t dx(s)$ is a stochastic process of bounded variation and $L(t) = \int_0^t dL(s)$ is deterministic and continuous (e.g., $dL$ does not include any Dirac delta functions).  This setting, and the results on the continuous-time Eisenberg-Noe model, can be extended to the case in which the assets and liabilities are additionally functions of the cash accounts $V$.  For simplicity, we will restrict ourselves so that the parameters are independent of the current cash accounts.

\begin{remark}
Though the external assets $x$ and marginal cash flows $dx(t)$ are typically assumed to be non-negative as shown in the static Eisenberg-Noe framework in the prior section, within the remainder of this work we relax that assumption except where explicitly mentioned otherwise as it is not necessary mathematically. 
\end{remark}

\begin{remark}
Though we introduce the liabilities to be due continuously in time, the formulations for this section can be derived as the limit of a discrete system (see Appendix~\ref{sec:discrete-top}).  Continuous-time formulations will be important in Section~\ref{sec:default} below in which defaults can occur outside of those specific payment dates due to capital (rather than liquidity) shortfalls.  Throughout this work, the discrete-time setting can be approximated by taking a twice-differentiable approximation of ``discrete'' liabilities so that $dL$ is a continuous and finite approximation of Dirac measures.
\end{remark}

\begin{remark}
Within Appendix~\ref{sec:discrete}, a representative balance sheet (Figure~\ref{fig:BalanceSheet}) is provided in a discrete-time setting. Though we consider continuous-time dynamics within the body of this work, the key concepts of inter-default time clearing are provided in the discrete-time setting. With this construction, we can see that the total external assets $x(t)$ are the aggregate of the cash-flows up to time $t$ and, similarly, the liabilities $L(t)$ are the aggregate of all obligations with maturity date $s \leq t$. In addition, and as highlighted mathematically below, any unpaid past liabilities will automatically be refinanced by the creditor on that debt. As discussed in Section~\ref{sec:fednow}, this setting corresponds with the stylized ruleset of the FedNow RTGS system.
\end{remark}

\begin{assumption}\label{ass:society}
The external assets $x: \bbt \to \bbr^{n+1}$ are an adapted process of bounded variation. 
The nominal liabilities $L: \bbt \to \bbr^{(n+1) \times (n+1)}_+$ are deterministic and twice differentiable; for notation we will define $dL(t) = \dot{L}(t) dt$ and $d^2L(t) = \ddot{L}(t) dt^2$.  Further, the relative liabilities to society is bounded from below by a level $\delta > 0$, i.e., $\inf_{t \in \bbt} \frac{dL_{i0}(t)}{\sum_{k \in \ncal_0} dL_{ik}(t)} \geq \delta > 0$ for all banks $i \in \ncal$.
\end{assumption}

As in the static Eisenberg-Noe system, we wish to consider the cash accounts $V: \bbt \to \bbr^{n+1}$ and the capital accounts $K: \bbt \to \bbr^{n+1}$, but now as functions of time.  We note that the introduction of time dynamics naturally leads to the separation of these notions which were identical in the static setting.  In this dynamic setting we can consider $V(t),K(t)$ to be a vector denoting the cash account and capital account (respectively) of all institutions in the financial system at time $t$.  The cash accounts grow and shrink due to the changes in incoming and outgoing cash flows whereas the capital accounts are only impacted by the accounting of the external assets.  In contrast to the static Eisenberg-Noe framework, herein we need to consider the results of the prior times to properly compute the system dynamics.  

\subsection{Continuous Time Evolution of the Capital Account}\label{sec:continuous-K}
First, we consider the capital account $K_i$ of each bank $i$ over time. The capital account encodes the difference between the banks total assets and liabilities. In order to accomplish this, we need to implement accounting rules to determine the valuation for any as-yet unpaid obligations and all future cash flows. We will follow mark-to-market accounting for the external assets $x$ and held-to-maturity or historical price accounting for the interbank obligations $L$. In particular, and for simplicity, we assume that no information is shared by banks and, thus, all obligations are marked in full on the bank balance sheets.\footnote{A more nuanced view of historical price accounting would include probabilities of default implied through, e.g., credit ratings.  As the modeling of credit ratings is beyond the scope of this work (and for simplicity), we only consider two states for accounting purposes: full payment or default.}  Historical price accounting is utilized for interbank obligations as these are non-marketable assets, i.e., there does not exist an open market on which they can be traded.\footnote{If a liquid market existed, interbank obligations could also be measured via mark-to-market pricing. Such considerations are beyond the scope of the current work.  This would require an extension of the network valuation adjustments as presented in, e.g.,~\cite{barucca2016valuation,BF18comonotonic}.}
The mark-to-market valuation of the external assets will be taken under the measure $\P$.  Notably, we do not require $\P$ to be the risk-neutral measure so that bank $i$ can invest in such a way to have positive (or negative) expected growth in the assets as we allow for $x_i(t) \neq \E_t[x_i(T)]$ where $\E_t[\cdot] := \E[\cdot | \fcal_t]$ is the conditional expectation w.r.t.\ the filtration $\fcal$.
Following this construction, we can encode the capital account $K_i(t)$ for each bank $i \in \ncal$ at time $t \in \bbt$ by:
\begin{equation}
\label{eq:K_no_default} K_i(t) = \E_t[x_i(T)] + \sum_{j\in\ncal} L_{ji}(T) - \sum_{j\in\ncal_0} L_{ij}(T).
\end{equation}

We wish to demonstrate the capital account under this setting by extending Example~\ref{ex:2bank-static} to a dynamic setting.
\begin{example}\label{ex:2bank-K}[Example~\ref{ex:2bank-static} continued]
Recall the $n = 2$ bank setting of Example~\ref{ex:2bank-static}. We now wish to extend this example to consider the case in which the external assets $x$ and liabilities $L$ are dynamic in time such that the terminal values of these balance sheet components correspond with the static values provided in that prior example, i.e., $x_1(T) = x_2(T) = 2.1$ and $L_{12}(T) = L_{10}(T) = 2$, $L_{21}(T) = 3$, and $L_{20}(T) = 1$. Due to the deterministic nature of this setting, the capital accounts are
\[K_1(t) = 1.1 \quad \text{and} \quad K_2(t) = 0.1\]
for every time $t \in [0,T]$. Note that because this setting incorporates no default events, the bank capitals exactly match those found in Example~\ref{ex:2bank-static}.
\end{example}

\subsection{Continuous Time Evolution of the Cash Account}

As in Section~\ref{sec:setting}, we will introduce the functional matrix $A: \bbt \to [0,1]^{(n+1) \times (n+1)}$ to be the \emph{relative exposure} matrix.  That is, $a_{ij}(t)V_i(t)^-$ provides the (negative) impact that firm $i$'s losses have on firm $j$'s cash account at time $t \in \bbt$.  For simplicity, and as mentioned above, we assume that all debts owed at any time have the same priority; that is, new obligations and old unpaid obligations are given the same priority and paid off proportionally from incoming assets.  With this concept, we can construct the relative exposure matrix akin to the approach in the static setting as provided in~\eqref{eq:static-A}.  Within this dynamic setting, the obligations from bank $i$ to bank $j$ at time $t$ is $dL_{ij}(t) + a_{ij}(t^-)V_i(t^-)^-$ where $a_{ij}(t^-)V_i(t^-)^-$ is the total amount of debts that roll forward; the total obligations of bank $i$ at time $t$ would similarly be given by $\sum_{k \in \ncal_0} dL_{ik}(t) + V_i(t^-)^-$ to account for any unpaid prior debts.  Assuming continuity of all such objects, we recover the relation
\begin{equation}\label{eq:dynamic-A}
da_{ij}(t)V_i(t)^- = dL_{ij}(t) - a_{ij}(t)\sum_{k \in \ncal_0} dL_{ik}(t)
\end{equation}
for every pair of banks $i,j \in \ncal_0$ and any time $t \in \bbt$.

With this setup, we can consider a differential system for the cash account $V$ and relative exposures $A$.
Briefly, at time $t$, the cash account $V(t)$ grows or shrinks according to the change in value of its primary business activities, i.e., the \emph{realized} cash flows, adjusted by the network multiplier $[I - A(t)^\T \Lambda(V(t))]^{-1}$ (as introduced in Section~\ref{sec:setting}) so as to capture the propagation of losses throughout the financial network.  In this way the change in the cash account follows comparably to the static Eisenberg-Noe fixed point problem as provided by the fictitious default algorithm (provided in Appendix~\ref{sec:FDA}).  The relative exposure matrix $A(t)$ changes over time in order to guarantee the appropriate equal priority of all obligations -- new obligations at time $t$ and all those unpaid up to time $t$ as encoded in~\eqref{eq:dynamic-A}. 
Mathematically, this is encoded in the differential system: 

\begin{align}
\label{eq:continuous-V} dV(t) &= [I - A(t)^\T \Lambda(V(t))]^{-1} \left(dx(t) - [I - A(t)^\T] dL(t) \vec{1}\right)\\
\label{eq:continuous-A} da_{ij}(t) &= \begin{cases} \frac{d^2L_{ij}(t) - a_{ij}(t)\sum_{k \in \ncal_0} d^2L_{ik}(t)}{\sum_{k \in \ncal_0} dL_{ik}(t)} &\text{if } i \in \ncal, \; V_i(t) \geq 0 \\ \frac{dL_{ij}(t) - a_{ij}(t)\sum_{k \in \ncal_0} dL_{ik}(t)}{V_i(t)^-} &\text{if } i \in \ncal, \; V_i(t) < 0 \\ 0 &\text{if } i = 0\end{cases} \quad \forall i,j \in \ncal_0
\end{align}
where $\Lambda(V) \in \{0,1\}^{(n+1)\times(n+1)}$ is the diagonal matrix of banks in delinquency, i.e.,
\[\Lambda_{ij}(V) = \begin{cases}1 &\text{if } i = j \neq 0 \text{ and } V_i < 0 \\ 0 &\text{else}\end{cases} \quad \forall i,j \in \ncal_0\]
and with initial conditions $V(0) \geq 0$ given and $a_{ij}(0) = \frac{dL_{ij}(0)}{\sum_{k \in \ncal_0} dL_{ik}(0)}\ind{i \neq 0} + \frac{1}{n}\ind{i = 0,\; j \neq 0}$ for all firms $i,j \in \ncal_0$.

\begin{remark}\label{rem:initialization}
All subsequent results permit the relaxation of the initial conditions so long as the relative exposures $A(0) \in \bba_0(V(0))$ where
\[\bba_0(v) := \left\{A \in [0,1]^{(n+1)\times(n+1)} \; \left| \; \begin{array}{l} A\vec{1} = \vec{1}, \; a_{ii} = 0, \; a_{i0} \geq \delta \; \forall i \in \ncal, \; a_{0j} = \frac{1}{n} \; \forall j \in \ncal, \\ a_{ij} = \frac{dL_{ij}(0)}{\sum_{k \in \ncal_0} dL_{ik}(0)} \; \forall i \in \ncal: v_i \geq 0, \, \forall j \in \ncal_0\end{array}\right.\right\}.\]
That is, all results follow so long as the dynamics begin from the solution to a fixed-point problem of the given initial conditions. For simplicity of exposition, and as we are interested in the dynamics that lead to distress events, we focus on the case in which banks begin with non-negative cash accounts.
\end{remark}

Recall from Proposition~\ref{prop:Leontief} that $I - A(t)^\T \Lambda(V(t))$ is invertible by standard input-output results.
The differential form in \eqref{eq:continuous-V} can be seen as a direct dynamic version of~\eqref{eq:EN-e-inv} in the static setting so that the cash account clears at every time.
The first case in \eqref{eq:continuous-A} is constructed by noting that $a_{ij}(t) = \frac{dL_{ij}(t)}{\sum_{k \in \ncal_0} dL_{ik}(t)}$ if $V_i(t) \geq 0$ and $i \in \ncal$ and $da_{0j}(t) = 0$ for any firm $j \in \ncal_0$ for all times $t$; the second case in \eqref{eq:continuous-A} follows directly from a reformulation of~\eqref{eq:dynamic-A}. 

\begin{remark}\label{rem:disc-to-cont}
The continuous-time differential system can, indeed, be constructed explicitly as the limit of the discrete-time setting with explicit time steps $\Dt$ (as defined in Appendix~\ref{sec:discrete}).  See Appendix~\ref{sec:disc-to-cont} for such a derivation. This derivation is the reason that the relative exposure matrix provides mathematical convenience as it allows us to directly study $V(t+\Dt)-V(t)$ which limits to $dV(t)$ for the cash account without having to explicitly account for limited liabilities.
\end{remark}

We will complete our discussion of the construction of this differential system by providing some properties on the relative liabilities and exposures matrix $A$. Notably, these properties are those that would be expected from the relative liabilities.  Namely, as a firm recovers from a delinquent state, its relative exposures return to be its instantaneous relative liabilities, that the relative exposures are bounded from below by 0 (and to society by $\delta$ as provided in Assumption~\ref{ass:society}), and the relative exposure matrix is row stochastic at all times. The proof of this proposition can be found in Appendix~\ref{sec:proof-propA}.
\begin{proposition}\label{prop:A}
Let $(dx,dL): \bbt \to \bbr^{n+1} \times \bbr^{(n+1) \times (n+1)}_+$ define a dynamic financial network satisfying Assumption~\ref{ass:society}.
Let $(V,A): \bbt \to \bbr^{n+1} \times \bbr^{(n+1) \times (n+1)}$ be any solution of the differential system \eqref{eq:continuous-V} and \eqref{eq:continuous-A} with initial condition $V(0) \in \bbr^{n+1}_{++}$.
The relative exposure matrix $A(t)$ satisfies the following properties:
\begin{enumerate}
\item\label{prop:A1} For any bank $i \in \ncal$, if $V_i(t) \nearrow 0$ as $t \nearrow \tau$ then $\lim_{t \nearrow \tau} a_{ij}(t) = \frac{dL_{ij}(\tau)}{\sum_{k \in \ncal_0} dL_{ik}(\tau)}$;
\item\label{prop:A2} For all times $t \in \bbt$ and for any bank $i \in \ncal$, the elements $a_{ij}(t) \geq 0$ for all banks $j \in \ncal$ and $a_{i0}(t) \geq \delta$; and
\item\label{prop:A3} For all times $t \in \bbt$ and for any bank $i \in \ncal_0$, the row sums $\sum_{k \in \ncal_0} a_{ik}(t) = 1$.
\end{enumerate}
\end{proposition}

With the differential construction of \eqref{eq:continuous-V} and \eqref{eq:continuous-A}, we seek to prove existence and uniqueness of the clearing solutions.  For notational simplicity, define the space of relative exposure matrices\footnote{The definition provided here, and used for the remainder of this work, is a subset of that given in Proposition~\ref{prop:Leontief} to enforce the obligations to society as stated in Assumption~\ref{ass:society}.}
\[\bba := \left\{A \in [0,1]^{(n+1) \times (n+1)} \; | \; A\vec{1} = \vec{1}, \; a_{ii} = 0, \; a_{i0} \geq \delta \; \forall i \in \ncal, \; a_{0j} = \frac{1}{n} \; \forall j \in \ncal\right\}.\]
From Proposition~\ref{prop:A}, we have already proven that if $(V,A): \bbt \to \bbr^{n+1} \times \bbr^{(n+1) \times (n+1)}$ is a solution to the continuous-time Eisenberg-Noe system then $A(t) \in \bba$ for all times $t \in \bbt$.
\begin{theorem}\label{thm:continuous}
Let $\bbt = [0,T]$ be a finite time period and let $(dx,dL): \bbt \to \bbr^{n+1} \times \bbr^{(n+1) \times (n+1)}_+$ define a dynamic financial network satisfying Assumption~\ref{ass:society}.
There exists a unique solution to the clearing cash accounts and relative exposures $(V,A)$ satisfying \eqref{eq:continuous-V} and \eqref{eq:continuous-A} if $V(0) \in \bbr^{n+1}_{++}$.
\end{theorem}
While the full proof is provided in Appendix~\ref{sec:proof-continuous}, the result follows from traditional arguments used in proving similar statements for jump-diffusion models.

We conclude this section by continuing our running example to demonstrate how the specific dynamics of the balance sheet of a bank can alter the realized path of the cash account.
\begin{example}\label{ex:2bank-V}[Example~\ref{ex:2bank-K} continued]
Consider the dynamic $n = 2$ bank setting of Example~\ref{ex:2bank-K}. Without loss of generality, we will assume $T = 1$. To demonstrate the impact of dynamic network effects on the cash account $V$, we will consider the external assets to be constant over time at $x_1(t) = x_2(t) = 2.1$ for all times $t \in [0,1]$.
We consider the network of obligations to vary over time so that bank 1 owes to bank 2 before it collects anything from bank 2, i.e.,
\[L_{12}(t) = 4t\ind{t \in [0,0.5]}, \quad L_{10}(t) = 2t, \quad L_{21}(t) = 6t\ind{t \in [0.5,1]}, \quad L_{20}(t) = t\]
with no other obligations.\footnote{Though these obligations are not twice differentiable at $t = 0.5$, we are able to find the unique solution up to this time and then restart the system from that point.} We wish to note that this network aggregates, at the terminal time $T$ to the setup described in Example~\ref{ex:2bank-K}.

For this system there are three stopping times that partition the interval $t \in [0,1]$: (i) the delinquency time of bank $1$ at $t = 0.35$; (ii) the change-over in obligations at time $t = 0.5$; and (iii) the recovery time of bank $1$ at $t = 0.725$. As such we can construct the relative exposures matrix $A(t)$ and cash accounts $V(t)$ as piecewise functions of time over the intervals $[0,0.35]$, $[0.35,0.5]$, $[0.5,0.725]$, and $[0.725,1]$.
\begin{itemize}
\item For time $t \in [0,0.35]$: During this interval, both banks 1 and 2 are solvent and liquid, and therefore are making their payments in full. Due to the linear nature of the obligations during this time, it is clear that the relative exposures are flat with $a_{12}(t) = 2/3$ and $a_{21}(t) = 0$ (with obligations to society implicitly defined from these). The clearing cash accounts follow exactly the linear form implied by the obligations with $V_1(t) = 2.1-6t$, $V_2(t) = 2.1+3t$, and $V_0(t) = 3t$.
\item For time $t \in [0.35,0.5]$: At time $t = 0.35$, the cash account of bank 1 hits 0 -- $V_1(0.35) = 0$ -- and this bank enters distress. This distress does not impact the relative exposures of $a_{12}(t) = 2/3$ and $a_{21}(t) = 0$, but means that the obligations from bank 1 do not accumulate to bank 2 or society. Thus the clearing cash accounts follow the updated formulas $V_1(t) = 2.1-6t$, $V_2(t) = 3.5-t$, and $V_0(t) = 0.7+t$. 
\item For time $t \in [0.5,0.725]$: Bank 2 begins paying obligations to bank 1 at time $t = 0.5$. These payments filter back from bank 1 to repay the unmet liabilities that accrued during $[0.35,0.5]$. Therefore any obligation that was paid to bank 1 is immediately sent as payments to bank 2 and society. 
Thus, immediately, we can conclude that bank 1's cash account needs to evolve as $V_1(t) = -2.9+4t \leq 0$. 
With this cash account, and knowing the initial relative exposure $a_{12}(0.5) = 2/3$, we solve the differential equation $da_{12}(t) = -2a_{12}(t)/(2.9-4t)$ to recover $a_{12}(t) = 4\sqrt{7.25-10t}/9$ (with $a_{10}(t)$ defined implicitly from this); notably, these relative exposures are decreasing as bank 1 pays off its debts to bank 2 while still accumulating further obligations to society.
In addition, $a_{21}(t) = 6/7$ by simple construction. With these values we can solve the differential system for the cash accounts of bank 2 and society:
\[V_2(t) = 7.1 - \frac{16\sqrt{10}}{9}(0.725-t)^{3/2} - 7t \quad \text{and} \quad V_0(t) = -2.9 + \frac{16\sqrt{10}}{9}(0.725-t)^{3/2} + 7t.\]
\item For time $t \in [0.725,1]$: At time $t = 0.725$, bank 1 has finally met all of its previously unpaid obligations and, therefore, exits distress. Therefore the evolution of the system returns to the linear construction to recover $V_1(t) = -2.9+4t$, $V_2(t) = 7.1-7t$, and $V_0(t) = 3t$.
\end{itemize}
The evolution of these cash accounts is provided in Figure~\ref{fig:2bank-V}.
We wish to note that the final values of these cash accounts $V(t) = (3,1.1,0.1)$ match the static cash accounts because all banks were solvent and liquid at the terminal time. We refer the interested reader to Appendix~\ref{sec:discussion-EN} for a discussion about when this dynamic model can replicate the static Eisenberg-Noe clearing solutions.
\begin{figure}[h]
\centering
\includegraphics[width=0.4\textwidth]{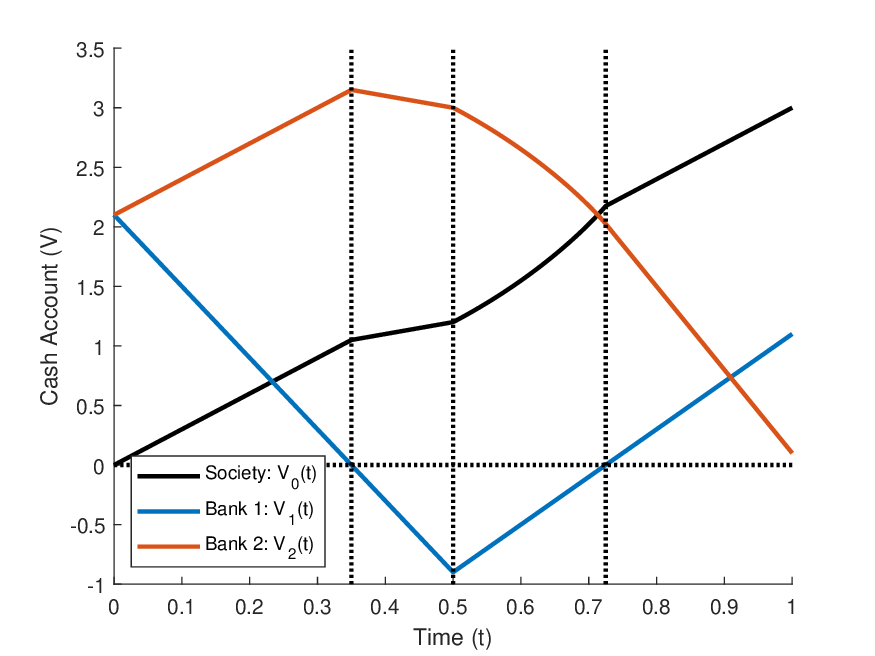}
\caption{Example~\ref{ex:2bank-V}: Clearing cash account over time with system transition times highlighted with dotted black lines.}
\label{fig:2bank-V}
\end{figure}
\end{example}

\section{Continuous-Time Clearing with Early Defaults}\label{sec:default}
\subsection{Model Formulation}\label{sec:default-model}

The dynamic framework presented thus far only considers the inter-default time system dynamics.  In this section we consider the default events and construct a holistic dynamic Eisenberg-Noe system.  Defaults in a \emph{dynamic} model can occur due to either illiquidity or insolvency (see Section~\ref{sec:default-rule} or, e.g.,~\cite{fs2019,feinstein2020dynamic}).  An illiquidity induced default for bank $i$ is caused by having insufficient cash-on-hand $V_i$.  An insolvency induced default for bank $i$ is caused by the capital account $K_i$ being too low. Because we are introducing defaults, the time-evolution of the $K_i$ requires additional dynamics beyond those presented in~\eqref{eq:K_no_default}; these are presented in~\eqref{eq:K}.  As we will show, these default events may occur either at positive capital levels due to capital regulations or nonpositive levels due to the incentive for shareholders to declare an early default in order to restructure to maximize long-term wealth.  As highlighted in Remark~\ref{rem:K=V_static}, in the static setting of~\cite{EN01,RV13,GK10} illiquidity and insolvency are identical notions.  As with the previous section, the proofs of the results herein are found in Appendix \ref{sec:continuous-proof}.

\begin{remark}\label{rem:continuous}
Though the cash account and liquidity questions can naturally be considered in discrete-time (as constructed in, e.g., Appendix~\ref{sec:discrete}), questions of solvency require a continuous-time construction as capital can drop below the required levels due to the value of the external investments or cash flow which can occur outside of dates at which obligations are due.
\end{remark}

It remains to consider the impacts that defaults have on the cash and capital accounts.
These defaults occur at stopping times $\tau$; specifically, let $\tau_i \in \lcal_{T}^{\infty}(\bbr_+)$ (i.e., a uniformly bounded $\fcal_T$-measurable random variable) denote the default time of bank $i$.  As detailed below, these defaults can be generated by liquidity (i.e., the cash account $V$) or solvency (i.e., the capital account $K$); for the present consideration we simply wish to consider that these are known stopping times.
Following the bankruptcy costs imposed in, e.g.,~\cite{RV13}, at default times, constant recovery rates of $\alpha,\beta,\gamma \in [0,1]$ are imposed on the liquid assets, the as-yet unpaid interbank assets from nondefaulted firms, and concurrent payments from simultaneous defaults respectively which impact both the system liquidity and system capital.
As highlighted in \cite{KV16}, and further studied in \cite{feinstein2020dynamic}, early defaults cause the cash account of banks to jump upward but the capital valuation to jump downward.  This occurs due to two countervailing effects: (i) liquidity increases as an early default involves an early repayment of future obligations with recovery rates $\alpha,\beta,\gamma$, but (ii) capital decreases due to the bankruptcy costs $1-\alpha,1-\beta,1-\gamma$ and historical price accounting of unpaid future obligations.
\begin{assumption}\label{ass:recovery-rates}
Throughout this work we will assume $\beta \geq \gamma$.  If this relation does not hold then it is possible that a default cascade can result in a greater fraction of outstanding obligations being repaid than if each default occurred in isolation; this could lead to cyclical default events as in~\cite{KV16}.  In fact, it is often desirable to assume $\alpha \geq \beta \geq \gamma$ due to the liquidity of these asset classes and the aforementioned ranking of $\beta,\gamma$.
\end{assumption}

We now wish to provide the mathematical formulation for the cash and capital accounts under default times $\tau$ and the financial rule described above.
First, the cash account $V$ follows the differential dynamics of~\eqref{eq:continuous-V} under the (as-yet) nondefaulted interbank obligations between default times, i.e.,
\begin{align*}
dV(t) &= \left[I - A(t)^\T\diag(\ind{\tau > t}) \Lambda(V(t))\right]^{-1}\left(dx(t) - [I - A(t)^\T\diag(\ind{\tau > t})]dL(t)\vec{1}\right).
\end{align*}
The cash account of bank $i$, however, has a jump effect at a default time of one of its counterparties; this occurs due to the early repayment of future obligations from bankruptcy settlement (and as detailed in, e.g.,~\cite{KV16}). This modification at the default time of bank $j$ is given by a joint fixed point problem between the cash account of the nondefaulting banks (i.e., banks $i$ such that $\tau_i > \tau_j$).
Notably, the cash account of any bank increases immediately following a default event due to, e.g., the auctioning of future assets for the defaulting firm to cover all of its unpaid (past and future) obligations; this is discussed in more detail in~\cite{CC15,KV16,feinstein2020dynamic}.
At the default time $\tau_i$, we need to consider the amount that any bank in this default cascade may pay.  These payments $P_i$ for bank $i$ are made at the default time $\tau_i$ following the recovery rates $\alpha,\beta,\gamma$ if the available assets are unable to satisfy the unpaid liabilities.  Specifically, these available assets are the sum total of the liquid holdings, the payments from other simultaneous defaults, and the \emph{recovered} value of the unpaid interbank assets from nondefaulted firms; the interbank assets are reduced by the recovery rate $\beta$ throughout as these are assumed to be illiquid and, as such, would not receive payment in full in, e.g., an auction (see~\cite{CC15}).  Mathematically, the clearing cash account fixed point problem at $\tau_j^+$ (i.e., the right-limit of time at $\tau_j$) is provided by the joint fixed point between $V_i(\tau_j^+)$ for $i$ such that $\tau_i > \tau_j$:
\begin{align*}
V_i(\tau_j^+) &= V_i(\tau_j) + \sum_{k \in \ncal} \left(\bar a_{ki}P_k \ind{\tau_k = \tau_j} + a_{ki}(\tau_j)\left[V_k(\tau_j)^- - V_k(\tau_j^+)^-\right]\ind{\tau_k > \tau_j} \right)\\
\bar a_{ki} &:= \frac{L_{ki}(T)- L_{ki}(\tau_k)+a_{ki}(\tau_k)V_k(\tau_k)^-}{\sum_{\ell \in \ncal_0} \left[L_{k\ell}(T)-L_{k\ell}(\tau_k)\right]+V_k(\tau_k)^-}
\end{align*}
and the payments $P_i$ made by defaulting firms $i$ such that $\tau_i = \tau_j$:
\begin{align*}
P_i &= \bar P_i \wedge \left[\alpha X_i + \beta F_i + \gamma\Psi_i(P,V(\tau_i^+))\right].
\end{align*}
Within the above fixed point problem, we set
\begin{align*}
\bar P_i &:= V_i(\tau_i)^- + \sum_{j\in\ncal_0}\left[L_{ij}(T) - L_{ij}(\tau_i)\right]\\
X_i &:= V_i(\tau_i)^+ + \E_t[x_i(T)] - x_i(\tau_i) \\ 
F_i &:= \sum_{j\in\ncal}\left[L_{ji}(T)-L_{ji}(\tau_i)+a_{ji}(\tau_i)V_j(\tau_i)^-\right]\ind{\tau_j > \tau_i} \\
\Psi_i(P,V(\tau_i^+)) &:= \sum_{k\in\ncal} \left(\bar a_{ki} P_k\ind{\tau_k = \tau_i} + a_{ki}(\tau_i)\left[V_k(\tau_i)^- - V_k(\tau_i^+)^-\right]\ind{\tau_k > \tau_i}\right).
\end{align*}
That is, $\bar P_i$ denotes bank $i$'s total unpaid liabilities at its default time $\tau_i$, $X_i$ denotes the value of bank $i$'s liquid assets at the time of default (i.e., the positive part of the cash account plus the expected cash flows), $F_i$ denotes the value of the unpaid interbank assets owed to bank $i$, and $\Psi_i(P,V(\tau_i^+))$ denotes bank $i$'s recovered liquidity due to payments at its default time through either simultaneous defaults or the recovery of unpaid old obligations from (as yet) undefaulted firms. Notably, as these payments are made by defaulting firms, the recovery rates $\alpha,\beta,\gamma$ are applied to all assets even if sufficient assets exist for full repayment of obligations.  As noted above in Assumption~\ref{ass:recovery-rates}, if $\gamma > \beta$ then it is possible for a default cascade to result in higher payments than if all such firms default in quick succession or even a cycle in which a cascaded default (see the definition below) rescues another defaulting firm.
The relative exposures at the default time update as well due to the possibility of covering previously unpaid obligations; specifically, $a_{ki}(\tau_j^+) = a_{ki}(\tau_j)$ if $V_k(\tau_j)V_k(\tau_j^+) > 0$ and $a_{ki}(\tau_j^+) = \frac{dL_{ki}(\tau_j)}{\sum_{\ell\in\ncal_0}dL_{k\ell}(\tau_j)}\ind{i\neq 0}+\frac{1}{n}\ind{i=0,\; j\neq 0}$ otherwise. 

\begin{proposition}\label{prop:cash}
Consider a financial network defined by $(dx,dL)$ satisfying Assumption~\ref{ass:society}.  Let $\alpha,\beta,\gamma \in [0,1]$ be the recovery rates.  Consider the default time $\tau_j$ of bank $j\in\ncal$.
There exists a unique clear cash account $V_i(\tau_j^+)$ for bank $i\in\ncal_0$ such that $\tau_i > \tau_j$ and payments $P_i$ for bank $i\in\ncal_0$ such that $\tau_i = \tau_j$.
\end{proposition}

Now, let us consider the capital account $K$.  Utilizing the accounting rules provided in Section~\ref{sec:continuous-K}, the capital account at time $t$ for each bank can be described explicitly without need for a differential system.  Specifically, for bank $i$, the capital account at time $t$ (given default times $\tau$) is given by:
\begin{align}
\label{eq:K} K_i(t) &= \E_t[x_i(T)] + \sum_{j\in\ncal} \left[L_{ji}(T)\ind{\tau_j > t} + \left(L_{ji}(\tau_j) - a_{ji}(\tau_j)V_j(\tau_j)^- + \bar a_{ji}P_j\right)\ind{\tau_j \leq t}\right] - \sum_{j\in\ncal_0} L_{ij}(T).
\end{align}
That is, the capital account is the sum of the accounting value of bank $i$'s cash flows from nondefaulting institutions and the realized payments from defaulting institutions.
Notably, in contrast to the cash account, the capital account of bank $i$ jumps downward at a default of some counterparty of bank $i$.  This occurs due to two effects: (i) historical price accounting causes an immediate drop in the value of interbank assets and (ii) the payments made from a defaulting institution is always bounded from above by its total obligations.

\begin{remark}\label{rem:cadlag}
We wish to note that the capital account $K$ is c\`adl\`ag as the accounting of defaults occurs at the default time.  However, the cash account is left-continuous with right limits at the default times as the transfer of payments from defaulting firms would take a clearing time to process.  This natural structure of the dynamic Eisenberg-Noe framework solves one of the primary difficulties of the discrete-time multiple maturity setting presented in~\cite{KV16} in which the default of one firm can cause another to recover from default; with the time delay such a sequence cannot occur.
\end{remark}

To finalize this dynamic framework, we wish to introduce the default times. 
Briefly, as a general setup, a bank will default if its cash and capital accounts are too low. Specifically, as provided in the ISDA Master Agreement and discussed in Section~\ref{sec:fednow}, if bank $i$'s capital account $K_i$ drops below zero (i.e., when the bank's liabilities exceed its assets in value) then bank $i$'s shareholders will declare default. Alternatively, if bank $i$ has failed to recover -- within a grace period $\Delta \geq 0$ -- from delinquency (based on its cash account $V_i$) then the bank will be declared to be in default by its creditors. 
That is, given the grace period $\Delta \geq 0$, bank $i$ defaults at stopping time $\tau_i$ defined by
\[\tau_i = \inf\{t \in \bbt \; | \; \min\{K_i(t) \, , \, \sup_{s \in [t-\Delta,t]} V_i(s)\} \leq 0\}.\]
Note that, by our construction of the financial system, $V_i(0) > 0$ for every bank $i$, therefore $\sup_{s \in [t-\Delta,t]} V_i(s) > 0$ if $t \in [0,\Delta]$.
If the grace period $\Delta = 0$ then any delinquency time is also a default time.

\subsection{Construction of an Optimal Clearing Solution}

With the early default model presented above, we now wish to consider the existence of a clearing cash and capital account.  Notably, though locally this problem is monotonic, we are unable to conclude globally that a greatest and least solution exist due to the possibility that some institution benefits from another declaring an early default (e.g., due to insolvency rather than illiquidity).  Additionally, uniqueness does \emph{not} hold generally due to the possibility of different sized default cascades. These difficulties are highlighted within Example~\ref{ex:nonunique} below.

\begin{corollary}\label{cor:defaults}
Let $\bbt = [0,T]$ be a finite time period and let $(dx,dL): \bbt \to \bbr^{n+1} \times \bbr^{(n+1)\times(n+1)}_+$ define a dynamic financial network satisfying Assumption~\ref{ass:society} with grace period $\Delta \geq 0$.
Let $\alpha,\beta,\gamma \in [0,1]$ be the recovery rates.
There exists a solution of the clearing cash and capital accounts $(V,K)$ with associated relative exposure matrix $A$.
\end{corollary}

We wish to demonstrate this clearing solution by continuing our running example from the prior sections to consider a brief discussion of the impact of the grace period $\Delta$ on the outcome of the financial system. 
\begin{example}\label{ex:2bank-default}[Example~\ref{ex:2bank-V}]
Consider the $n = 2$ bank setting of Example~\ref{ex:2bank-V}. For simplicity, herein we fix the recovery rates $\alpha,\beta,\gamma = 0$ so that no assets are recovered during default as taken in, e.g., \cite{GK10}. Following the optimal clearing solution construction as in the proof of Corollary~\ref{cor:defaults}, there are two key grace periods $\Delta = \frac{3}{8}(1 - \frac{1}{5}[\frac{3}{4}]^{1/3}) \approx 0.3069$ and $\Delta = 0.375$. Recall that, up until a default, this system will evolve exactly as in Examples~\ref{ex:2bank-K} and~\ref{ex:2bank-V}.
\begin{itemize}
\item If $\Delta \in (0.375 , \infty)$ then there are no defaults within this system as bank 1 has enough time to recover from distress. As such the cash accounts follow exactly as in Example~\ref{ex:2bank-V} and the capital accounts are constant over time: $K_1(t) = 1.1$, $K_2(t) = 0.1$, and $K_0(t) = 3$ for all times $t$ as in Example~\ref{ex:2bank-K}.
\item If $\Delta \in (0.3069 , 0.375]$ then bank 1 defaults at time $\tau = 0.35+\Delta$, but bank 2 recovers enough to remain solvent. Setting the capital to 0 once the firm defaults, we find that the clearing capital accounts follow the piecewise constant forms: $K_1(t) = 1.1 \times \ind{t < \tau}$, $K_2(t) = 0.1 \times \ind{t < \tau} + [-1.9 + 4\min\{\tau,0.5\} + a_{12}(\tau)V_1(\tau)]\ind{t \geq \tau}$, and $K_0(t) = 3\times\ind{t < \tau} + [1.7 + 2\Delta + a_{10}(\tau) V_1(\tau)]\ind{t \geq \tau}$. 
The clearing cash accounts for bank 2 and society are shifted downward by a change of slope so that $V_2(t)$ for $t > \tau$ has a slope of $-7$ and $V_0(t)$ has a slope of $1$.
\item If $\Delta \in [0 , 0.3069]$ then both banks 1 and 2 default at time $\tau = 0.35+\Delta$. Following the same convention as above that capital is 0 after default, $K_1(t) = 1.1 \times \ind{t < \tau}$ and $K_2(t) = 0.1 \times \ind{t < \tau}$. The capital for society is given by the piecewise constant function $K_0(t) = 3\times\ind{t < \tau} + [3\tau + a_{10}(\tau)V_1(\tau)]\ind{t \geq \tau}$. The clearing cash account for society is updated after $\tau$ so that $V_0(t) = V_0(\tau)$ for all $t > \tau$.
\end{itemize}
In Figure~\ref{fig:2bank-default}, we plot the terminal capital levels for banks 1 and 2 as well as society as a function of the grace period $\Delta$. Note that the capital account of society accelerates after $\Delta = 0.15$ as this corresponds to time $\tau = 0.5$ when bank 1 is able to start repaying its missed obligations. Furthermore, we wish to mention that these terminal capital accounts are equivalent to the terminal cash accounts under the same scenarios as there are no more future obligations to distinguish the types of holdings.
\begin{figure}[h]
\centering
\includegraphics[width=0.4\textwidth]{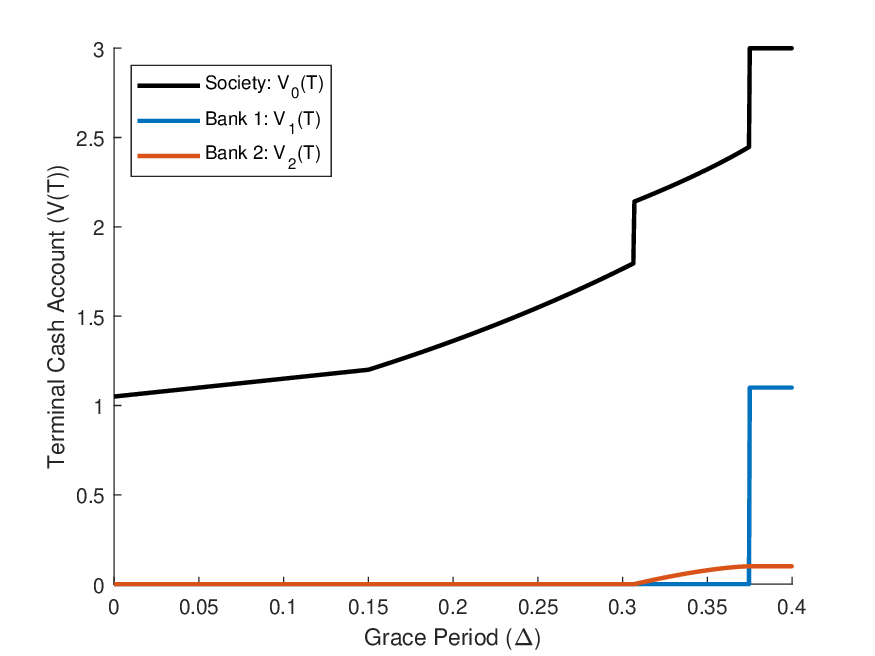}
\caption{Example~\ref{ex:2bank-default}: Terminal capital account as a function of the grace period with system transition points highlighted with dotted black lines.}
\label{fig:2bank-default}
\end{figure}
\end{example}

\begin{assumption}\label{ass:max}
For the remainder of this work, we will solely focus on the constructed clearing solution presented within the proof of Corollary~\ref{cor:defaults} as it -- locally -- results in the smallest possible default cascade at all default times.
That is, conceptually, we follow a fictitious default approach (similar to that taken in~\cite{EN01,RV13}) so that, forward in time, we assume there are no new defaults and verify if that holds; if verification fails then a default time is triggered and the procedure presented in Proposition~\ref{prop:cash} is utilized.
\end{assumption}

We wish to conclude our discussion of the optimal clearing solution by providing an example that elaborates on the local -- but \emph{not} global -- optimality of the constructed equilibrium. Specifically, we construct a system with a second clearing solution that is not dominated by the ``optimal'' one constructed through the algorithm provided in the proof of Corollary~\ref{cor:defaults}.
\begin{example}\label{ex:nonunique}
Consider a dynamic $n = 3$ bank system with society node over time $[0,1]$. For this example we consider recovery rates $\alpha = 0.75$, $\beta = 0.5$, and $\gamma = 0$ on liabilities:
\[L(t) = \left(\begin{array}{cccc} t & 0 & f_{(a,b)}(t/0.7) & f_{(a,b)}((t-0.7)/0.3) \\
    2t & f_{(a,b)}((t-0.7)/0.3) & 0 & 0 \\
    t & f_{(a,b)}((t-0.7)/0.3)/2 & f_{(a,b)}((t-0.7)/0.3)/2 & 0\end{array}\right)\]
where $f_{(a,b)}: \bbr \to [0,1]$ is such that $t \in [0,1] \mapsto f_{(a,b)}(t)$ is the regularized incomplete Beta function with parameters $a > 2, b > 2$ while $f(t) = 0$ for $t < 0$ and $f(t) = 1$ for $t > 1$.
Further, consider the external assets to be given by $x(t) = (1.7 , 1.7 , 1.01)^\T$ for all times $t \in [0,1]$. Finally, we set the grace period $\Delta = 0$ to simplify the discussion. We now wish to consider two clearing solutions:
\begin{itemize}
\item Following algorithm within the proof of Corollary~\ref{cor:defaults}, bank 1 will default at time $0.7$ from its cash account $V_1(0.7) = 0$. This default will cause contagion in which bank 3 will fail as its capital account will drop from $0.1$ down to $-0.4131$. Consequently, the capital for bank 2 drops to $0.7894$ and survives throughout the period. The external society node will end the period with $2.969$.
\item Consider, now the counterfactual setting in which banks 1 and 2 are both in default at the initial time $t = 0$. Following the clearing equations, we then find capital accounts $K(t) \approx (2.5250 \, , \, -0.2917 \, , \, -0.2917 \, , \, 0.0267)^\T$ for all times $t$. Notably, this is a fixed point and no other defaults would occur throughout the time period.
\end{itemize}
Notably, bank 3 is better off under the second setting in which banks 1 and 2 both default at initialization compared to the ``locally optimal'' algorithm provided in the proof of Corollary~\ref{cor:defaults}. Bank 2 and society have the contrasting view on these solutions and find the algorithm optimal.
\end{example}

\subsection{Illiquidity and Insolvency Conditions}\label{sec:default-rule}
Within this section we wish to, first, provide a mathematical description of the cause of a default, i.e., due to illiquidity (the cash account $V$) or insolvency (the capital account $K$). Second, we want to classify the first \emph{distress} event as being due to delinquency or insolvency due to the dynamics of the external assets $x$.

\begin{definition}\label{defn:defaults}
We classify defaults into three cases.  Let $i \in \ncal$ be an arbitrary bank.  Fix $\omega\in\Omega$ such that $\tau_i(\omega) \leq T$, i.e., such that the default of bank $i$ occurs during the time period of interest.  This default is classified as
\begin{enumerate}
\item an \textbf{\emph{illiquidity}} if bank $i$ is unable to recover from a delinquency within the grace period $\Delta$, i.e., $\lim_{t \nearrow \tau_i(\omega)} \sup_{s \in [t-\Delta,t]} V_i(s,\omega) = 0$; 
\item an \textbf{\emph{insolvency}} if bank $i$'s shareholders declare bankruptcy due to to an anticipated drop in capital, i.e., $\lim_{t \nearrow \tau_i(\omega)} K_i(t,\omega) = 0$; and
\item a \textbf{\emph{cascaded default}} if bank $i$ defaults solely due to the failure of other institutions, i.e., $\lim_{t \nearrow \tau_i(\omega)} \min\{K_i(t,\omega) \, , \, \sup_{s \in [t-\Delta,t]} V_i(s,\omega)\} > 0$. 
\end{enumerate}
\end{definition}

Before discussing the first distress time (i.e., the time of the first default or delinquency), we wish to provide some quick remarks on how each of these types of default can be prevented by, e.g., a central bank. 
\begin{remark}\label{rem:defaults} \hspace{1mm}
\begin{itemize}
\item An illiquidity occurs when the default is entirely triggered by the cash account, i.e., $\tau_i(\omega) = \inf\{t \in \bbt \; | \; \sup_{s \in [t-\Delta,t]} V_i(s,\omega) \leq 0\}$. As such, an illiquid bank can be rescued by providing a loan with sufficiently low interest rate (so as to not reduce the capital account too much). That is, illiquidities can be prevented by a lender of last resort.
\item In contrast, an insolvency occurs when the default is entirely triggered by the capital account, i.e., $\tau_i(\omega) = \inf\{t \in \bbt \; | \; K_i(t,\omega) \leq 0\}$. As such, an insolvency requires a bailout to raise the equity of the firm.
\item We consider cascaded defaults separate from the other two cases as these defaults can be prevented solely by providing a necessary loan or bailout for some other firm(s) depending on the type of initial default(s) occurring.  Due to the characteristics of the cash account $V_i$ at a default time, a cascaded default is a type of insolvency; that is, though a cascade can be triggered by either an illiquidity or insolvency, all subsequent defaults at that time must be due to the immediate impacts to the capital account.
\item If $\alpha = \beta=\gamma=1$, i.e.\ if there were no bankruptcy costs, then an insolvency cannot trigger a default cascade as the historical price accounting becomes the realized value.  However, in this situation, a liquidity event can still trigger a cascade.
\end{itemize}
\end{remark}

By construction, the first distress to occur in the financial system must be either a delinquency or an insolvency.  The following proposition will provide a simple framework that guarantees the type of this first default based on the dynamics of the obliged cash flows.  In particular, this first distress time is a delinquency (i.e., has a cash account $V(t)<0$) if the system is driven by submartingales and an insolvency (i.e., has capital account $K(t) = 0$) if the system is driven by supermartingales. As a stress scenario can generally be described by a supermartingale, this result indicates that the capital level of each bank would be the key element for regulators to watch in the midst of a financial crisis. However, in contrast, when the market is trending upward (i.e., a submartingale), then a systemic event can still be triggered via the delinquency of a firm.
\begin{proposition}\label{prop:defaults}
Consider the setting of Corollary~\ref{cor:defaults} with fixed grace period $\Delta \geq 0$.  If the stochastic process $\{x(t)+L(t)^\T\vec{1}-L(t)\vec{1} \; | \; t \in \bbt\}$ is a:
\begin{enumerate}
\item \emph{supermartingale} then the first distress is an insolvency; 
\item \emph{submartingale} then the first distress is a delinquency;  
\item \emph{martingale} then the cash account and capital account coincide exactly up to the first possible distress time (i.e., $K(t) = V(t)$ for any time $t$ where both are non-negative). 
\end{enumerate}
\end{proposition}
\begin{proof}
These results follow trivially by noting that, prior to the first distress, $K(t) = \E_t[x(T) + L(T)^\T\vec{1} - L(T)\vec{1}]$ and $V(t) \leq x(t) + L(t)^\T\vec{1} - L(t)\vec{1}$ with equality if $V(t) \geq 0$.
\end{proof}

\section{Numerical Examples}\label{sec:numerics}

Within this section we wish to consider two simple numerical examples to demonstrate certain properties of the above model.  First, in Example~\ref{ex:sample}, we consider a single sample path of the clearing cash and capital accounts so as to visualize the effects of early defaults. In particular, the jumps in both the capital account (downward) and the cash account (upward) at default times are clearly demonstrated. As such, the default contagion is clear along this sample path.  Then, in Example~\ref{ex:defaults}, we study the impact of the grace period $\Delta$ on the health of the financial system as measured by the default times of the different institutions in the system. Notably, the system dynamics endogenously recover non-monotonic outcomes for the health of the financial system as the grace period changes.
\begin{assumption}\label{ass:gbm}
For simplicity, throughout this section, let the external assets $x$ follow a discrete approximation of geometric Brownian motions without drift, i.e.,
\[dx_i(t) =  x_i(t) \sigma_i dB_i(t) \]
for constants $\sigma \in \bbr^{n+1}_+$ with (correlated) Brownian motions $B_1,...,B_{n+1}$.
\end{assumption}

\begin{example}\label{ex:sample}
Consider a financial system with three banks, each with an obligation to an external societal node.  Consider the time interval $\bbt = [0,1]$. Let these banks have obligations over time
$L_{ij}(t) = t \ind{i \neq j}$
for $i \in \ncal$ and $j \in \ncal_0$, i.e., a symmetric and completely connected system of obligations.  This assumption of constant relative liabilities is shown in Appendix~\ref{sec:discussion-EN} to provide the Eisenberg-Noe clearing wealths at the terminal time when neglecting defaults; herein we will see that with early defaults this no longer holds.  We will take the external assets to follow Assumption~\ref{ass:gbm} with initial value $x_i(0) = 2$, volatility $\sigma_i = 1$, and pairwise correlations of the Brownian motions $\rho_{ij} = \ind{i = j} + \frac{1}{2}\ind{i \neq j}$ for any pair of banks $i,j \in \ncal_0$.  As the purpose of this example is to demonstrate the effects of defaults on a single sample path of the cash and capital accounts, we will take the simplified system without grace period $\Delta = 0$, i.e., a bank defaults if its cash or capital account reaches $0$.  

First, recall that external assets $x$ are constructed to be a martingale in Assumption~\ref{ass:gbm}. Therefore, by the symmetry of the interbank assets but with the strictly positive external liabilities, the provided setting is such that $\{x(t) + L(t)^\T\vec{1} - L(t)\vec{1} \; | \; t \in \bbt\}$ is a supermartingale. Therefore, as provided in Proposition~\ref{prop:defaults}, the first default will occur due to an insolvency (i.e., $\lim_{t \nearrow \tau_{[1]}} K_{[1]}(t) = 0$). 
(In fact, as seen in Figure~\ref{fig:sample}, for this sample path the choice of grace period $\Delta$ is irrelevant as a delinquency does not occur.)
At default times, we take the recovery rates as $\alpha = \frac{1}{2}$, $\beta = \frac{1}{4}$, and $\gamma = \frac{1}{5}$.
A single sample path of this clearing system is provided in Figure~\ref{fig:sample}.  As can be seen in this sample path, bank $3$ has an insolvency event at time $\tau_3 \approx 0.224$ triggering a bank $2$ cascaded default.  The impact to bank $1$ is clear, and though not a direct cascaded default, the impact to the capital of bank $1$ contributes to its early default as well.  The cash accounts of society and bank $1$, however, improve in the immediate aftermath of the defaults at time $t \approx 0.224$; society's cash account clearly slows its rate after the defaults so that though the cash account improves in the short-term, it is harmed in the long-term.
\begin{figure}[t]
\centering
\begin{subfigure}[t]{0.45\linewidth}
\centering
\includegraphics[width=\linewidth]{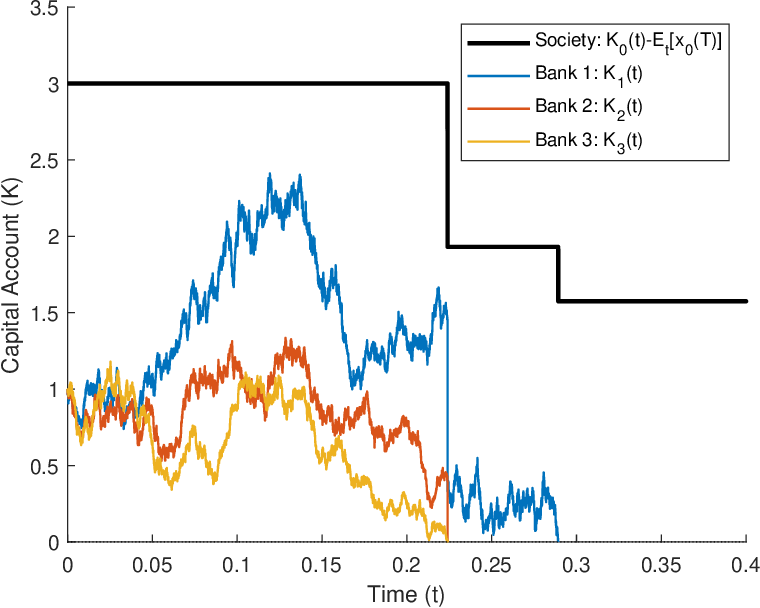}
\caption{Example~\ref{ex:sample}: Clearing capital accounts over time.}
\label{fig:sample-K}
\end{subfigure}
~
\begin{subfigure}[t]{0.45\linewidth}
\centering
\includegraphics[width=\linewidth]{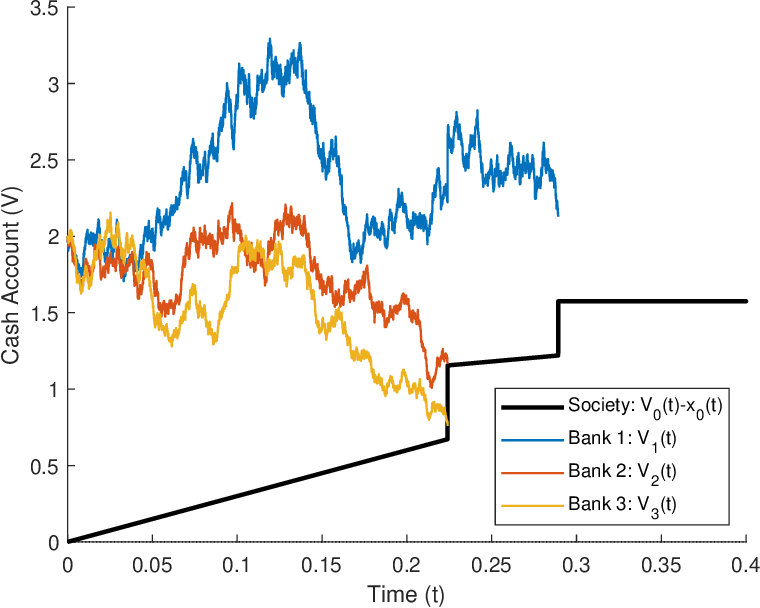}
\caption{Example~\ref{ex:sample}: Clearing cash accounts over time.}
\label{fig:sample-V}
\end{subfigure}
\caption{Example~\ref{ex:sample}: A single sample path of the clearing cash and capital accounts to demonstrate the impacts of early defaults on the wealth of other institutions.}
\label{fig:sample}
\end{figure}
\end{example}

\begin{example}\label{ex:defaults}
Consider again a three bank plus societal node system; in comparison to the network of Example~\ref{ex:sample}, herein we consider minimal societal obligations so as to more clearly visualize default contagion within the system.
Consider the time interval $\bbt = [0,1]$ and let banks have time-varying obligations: $L_{1j}(t) = t^2\ind{j \neq 1}$ for $j \in \{2,3\}$ and $L_{10}(t) = 10^{-3}t^2$; $L_{2j}(t) = t\ind{j \neq 2}$ for $j \in \{1,3\}$ and $L_{20}(t) = 10^{-3}t$; and $L_{3j}(t) = 2t(1-0.5t)\ind{j \neq 3}$ for $j \in \{1,2\}$ and $L_{30}(t) = 2\times10^{-3}t(1-0.5t)$. 
We assume, additionally, that the external assets follow Assumption~\ref{ass:gbm} with initial value $x_i(0) = 1$, volatility $\sigma_i = 1$, and pairwise correlations of the Brownian motions $\rho_{ij} = \ind{i = j} + \frac{1}{2}\ind{i \neq j}$ for any pair of banks $i,j \in \ncal_0$.  
At default times, we take the recovery rates as $\alpha = \frac{1}{2}$, $\beta = \frac{1}{4}$, and $\gamma = \frac{1}{5}$ for demonstration purposes.
The purpose of this example is to demonstrate the effects of the grace period $\Delta$ on system health as measured by the default times of each firm.  We investigate these impacts by repeatedly studying the same single sample path under varying grace periods $\Delta$ so that the effects of changing the default rule become evident. 

The results of changing the grace period are displayed in Figure~\ref{fig:defaults}.  First, in Figures~\ref{fig:sample-VK-D=0}-\ref{fig:sample-VK-D=10}, we consider three sample paths directly with grace period $\Delta = 0,~0.05,~0.1$ respectively. Notably, when $\Delta = 0$, bank 3 has an illiquidity default at $\tau_3 \approx 0.143$; this illiquidity causes a default cascade as the capital account for both banks 1 and 2 are immediately driven below 0.  
In comparison, when $\Delta = 0.05$, bank 3 is able to survive beyond a few early periods of delinquency, but it still defaults at $\tau_3 \approx 0.501$.  Under this setting with a grace period, both banks 1 and 2 have sufficient capital accumulated to withstand the immediate drop in capital; however, these lower capital values work contribute to the insolvency of bank 1 at $\tau_1 \approx 0.766$ and, as a consequence, the cascaded default of bank 2 at $\tau_2 = \tau_1$.  
Intriguingly, when we increase the grace period further to $\Delta = 0.1$, though bank 3 survives longer ($\tau_3 \approx 0.620$), banks 1 and 2 default earlier than in the $\Delta = 0.05$ case (with $\tau_1 = \tau_2 \approx 0.734$).  This indicates that there are strong nontrivial relations between the health of the financial system and the grace period $\Delta$.
This is even more clearly demonstrated in Figure~\ref{fig:tau} where the default times are directly plotted.  Though bank 3 is uniformly delaying its default time as $\Delta$ increases, banks 1 and 2 have default times that fluctuate in a chaotic manner.  Of particular interest, when the grace period $\Delta \in (0.096,0.099)$, bank 2 is able to survived and did not default at all even with the default of bank 1; however at $\Delta \geq 0.099$, bank 2 again is stressed in such a way that it defaults.
The dynamics in the default times for banks 1 and 2 are significantly driven by the actual cash account of bank 3 at its default time. Specifically, there are two countervailing forces to the health of banks 1 and 2 based on the grace period for bank 3: (i) a longer grace period delays the time when these (otherwise solvent) institutions need to write down their capital accounts $K$ and (ii) a longer grace period can result in the bank 3 having even larger losses to transmit through the financial system. In particular, due to the random evolution of the cash account of bank 3, these two countervailing effects fluctuate in importance.
As such, we are left to conclude that the grace period $\Delta$ can have significant systemic risk implications, though quantifying these effects is highly nontrivial and requires careful simulation.
\begin{figure}[t]
\centering
\begin{subfigure}[t]{0.45\linewidth}
\centering
\includegraphics[width=\linewidth]{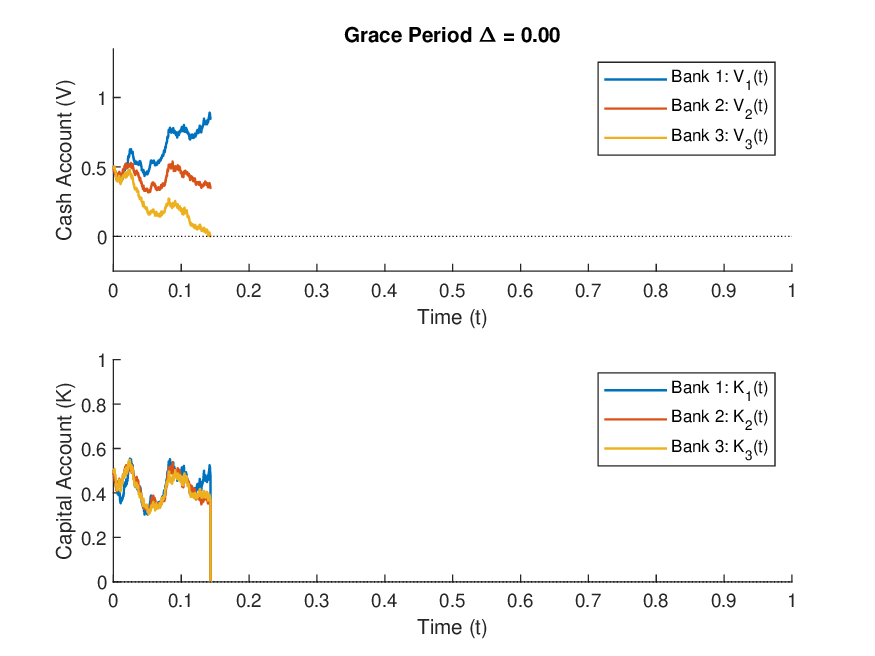}
\caption{Example~\ref{ex:defaults}: Cash and capital accounts for $\Delta = 0$.}
\label{fig:sample-VK-D=0}
\end{subfigure}
~
\begin{subfigure}[t]{0.45\linewidth}
\centering
\includegraphics[width=\linewidth]{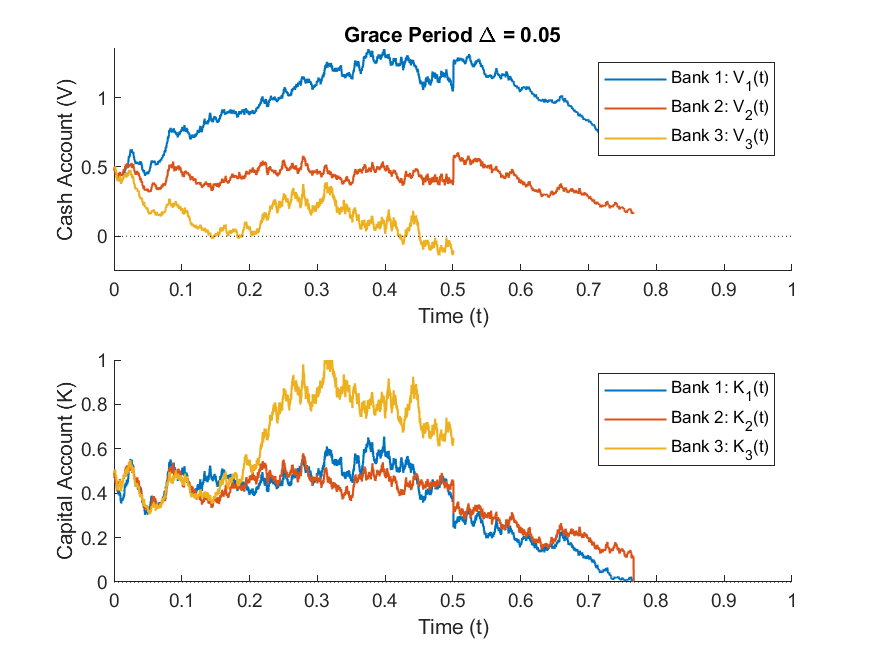}
\caption{Example~\ref{ex:defaults}: Cash and capital accounts for $\Delta = 0.05$.}
\label{fig:sample-VK-D=5}
\end{subfigure}
~\\
\begin{subfigure}[t]{0.45\linewidth}
\centering
\includegraphics[width=\linewidth]{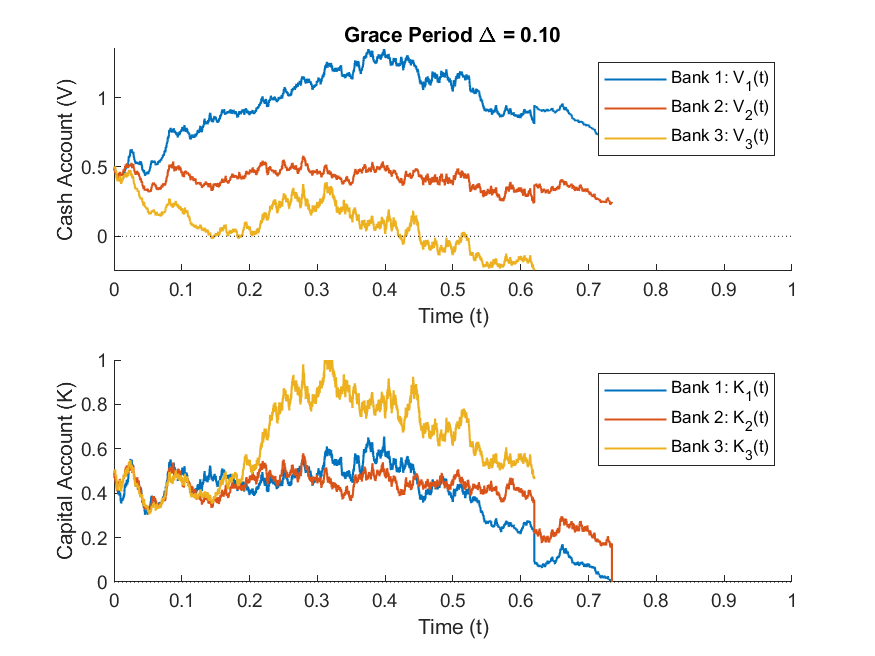}
\caption{Example~\ref{ex:defaults}: Cash and capital accounts for $\Delta = 0.1$.}
\label{fig:sample-VK-D=10}
\end{subfigure}
~
\begin{subfigure}[t]{0.45\linewidth}
\centering
\includegraphics[width=\linewidth]{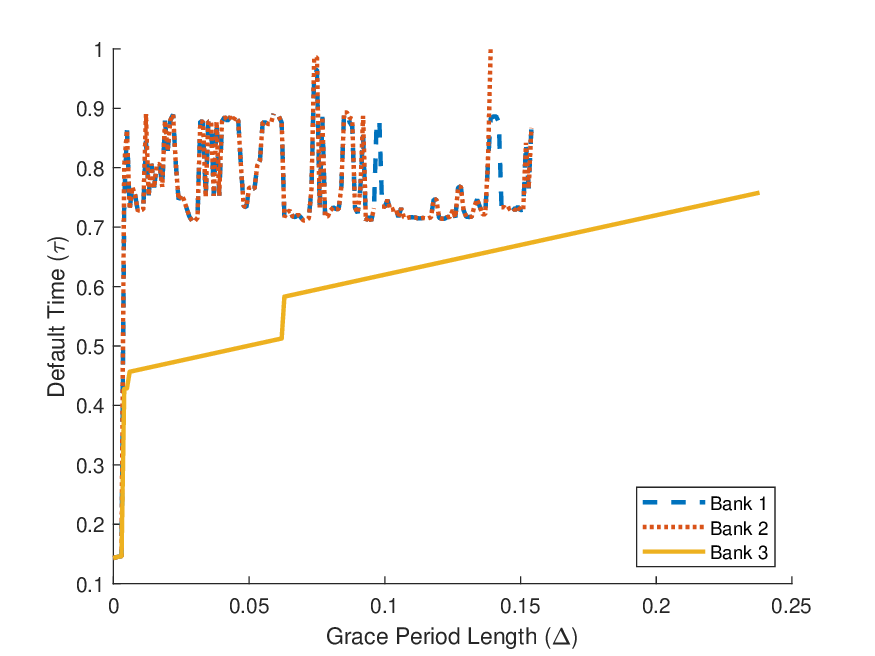}
\caption{Example~\ref{ex:defaults}: Default times as a function of the grace period.}
\label{fig:tau}
\end{subfigure}
\caption{Example~\ref{ex:defaults}: Depictions of the health of the financial system under a single sample path with varying a grace period.}
\label{fig:defaults}
\end{figure}
\end{example}

\section{Conclusion}\label{sec:conclusion}
In this paper we considered an extension of the financial contagion model of \cite{EN01,RV13} to allow for assets and obligations to be dynamic in time.  We presented this model in a continuous-time framework and provide conditions for existence and uniqueness of the clearing solutions under deterministic and stochastic settings.  In this dynamic setting, we introduced mathematical definitions for delinquency as well as illiquidity and insolvency defaults. Notably, the distinctions between delinquency and defaults (as well as between default types) can only occur due to the time dynamics as introduced within this work; in the static Eisenberg-Noe system, the capital and cash accounts are identical and no grace period can be included to permit recovery from a missed payment.

Five clear extensions of this model are apparent to us, and which we foresee creating further divergence between static and dynamic models.  
The first extension is to update the accounting valuation of interbank obligations from the backwards-looking historical price accounting procedure used here to a forwards-looking mark-to-market accounting rule.  Such a network valuation adjustment procedure has been studied in the single period setting in~\cite{barucca2016valuation,BF18comonotonic} and on a discrete tree in~\cite{feinstein2022endogenous}.
The second extension is the inclusion of illiquid assets and fire sales.  In the static models, e.g., \cite{CFS05,AFM16,feinstein2015illiquid}, there is no first mover advantage to liquidating assets as all firms receive the same price.  However, in a dynamic model there may be advantage to liquidating early in order to receive a higher price, but which may precipitate a larger fire sale amongst the other firms as is undertaken in a pure price-mediate contagion setting in \cite{feinstein2020capital}.  
The third extension is the inclusion of contingent payments and credit default swaps.  In the static setting this has recently been considered by \cite{banerjee2017insurance,SSB16,SSB16b}.  By considering the network dynamics to be dependent on the history of clearing cash accounts, many of the difficulties reported in the static works are likely to be resolved naturally; we refer to \cite{banerjee2017insurance} which provides an initial discussion of this extension.  
The fourth extension is the inclusion of margin requirements and collateralized obligations.  This has been explored in the static setting by \cite{ghamami2021collateralized}.
The final extension, for which we believe the proposed dynamic model will be especially useful, is in considering strategic or dynamic actions by the market participants, e.g., incorporating strategic decisions on rolling forward of debt and early defaults akin to that studied in~\cite{cossin2007credit}.


\bibliographystyle{apalike}
\bibliography{bibtex2}

\newpage
\appendix

\section{Proof of Results in Sections~\ref{sec:continuous} and ~\ref{sec:default}}\label{sec:continuous-proof}

\subsection{Proof of Proposition~\ref{prop:A}}\label{sec:proof-propA}

\begin{proof}
\begin{enumerate}
\item Consider firm $i \in \ncal$. By assumption we have that $a_{ij}(t)$ for $t \nearrow \tau$ solves the first order differential equation:
\begin{equation*}
\frac{da_{ij}(t)}{dt}+\frac{\sum_{k \in \ncal_0} dL_{ik}(t)/dt}{V_i(t)^-}a_{ij}(t)=\frac{dL_{ij}(t)/dt}{V_i(t)^-}.
\end{equation*}
For sake of simplicity, let this differential equation start at time $0$ with $V_i(0) < 0$ and some initial value $a_{ij}(0)$ satisfying the conditions of Remark~\ref{rem:initialization}.  Then this differential equation can be solved via the integrating factor $\nu(t) := \int_0^t \frac{\sum_{k \in \ncal_0} dL_{ik}(s)}{V_i(s)^-}ds $.
Thus for $t \nearrow \tau$ it follow that
\begin{equation*}
a_{ij}(t)= e^{-\nu(t)}\left[\int_0^t e^{\nu(s)} \frac{dL_{ij}(s)}{V_i(s)^-} + a_{ij}(0)\right].
\end{equation*}
Therefore, utilizing L'H\^ospital's rule,
\begin{align*}
\lim_{t \nearrow \tau} a_{ij}(t) &= \lim_{t \nearrow \tau} e^{-\nu(t)}\left[\int_0^t e^{\nu(s)} \frac{dL_{ij}(s)}{V_i(s)^-} + a_{ij}(0)\right]\\
&= \lim_{t \nearrow \tau} \frac{e^{\nu(t)}\frac{dL_{ij}(t)}{V_i(t)^-}}{e^{\nu(t)} \frac{d}{dt} \nu(t)} = \lim_{t \nearrow \tau} \frac{dL_{ij}(t)/V_i(t)^-}{\sum_{k \in \ncal_0} dL_{ik}(t)/V_i(t)^-}
= \frac{dL_{ij}(\tau)}{\sum_{k \in \ncal_0} dL_{ik}(\tau)}.
\end{align*}

\item First, if $V_i(t) \geq 0$ then by construction (and the above result) it follows that $a_{ij}(t) = \frac{dL_{ij}(t)}{\sum_{k \in \ncal_0} dL_{ik}(t)} \geq 0$ for any $i,j \in \ncal_0$ and $a_{i0}(t) \geq \delta$ by this construction.  Consider now the case for $V_i(t) < 0$ and assume $a_{ij}(t) < 0$.  Let $\tau = \sup\{s \leq t \; | \; V_i(s) = 0\}$.  Since $a_{ij}(\tau) \in [0,1]$ by construction and the relative exposures are continuous, this implies there exists some time $s \in [\tau,t)$ such that $a_{ij}(s) = 0$.  By the definition of the relative exposures, this must follow that $da_{ij}(s) \geq 0$ for any time $a_{ij}(s) \leq 0$ (with $da_{ij}(s) > 0$ if $a_{ij}(s) < 0$), thus $a_{ij}(t) < 0$ can never be reached.  Further, assume $a_{i0}(t) < \delta$.  By Assumption~\ref{ass:society}, if $a_{i0}(s) \leq \frac{dL_{i0}(s)}{\sum_{k \in \ncal_0} dL_{ik}(s)}$ then $da_{i0}(s) \geq 0$.  In particular, if $a_{i0}(s) \leq \delta$ then $da_{i0}(s) \geq 0$ (with $da_{i0}(s) > 0$ if $a_{i0}(s) < \delta$).  Thus, by the same contradiction found in the case for $j \in \ncal$, we are able to bound $a_{i0}(t) \geq \delta$.

\item First, if $i = 0$ then $\sum_{j \in \ncal_0} a_{0j}(t) = 1$ by property that $a_{0j}(t) = \frac{1}{n}\ind{j \neq 0}$ for all times $t$.  Now consider $i \in \ncal$, if $V_i(t) \geq 0$ then by construction (and the above result) it follows that $\sum_{j \in \ncal_0} a_{ij}(t) = \sum_{j \in \ncal_0} \frac{dL_{ij}(t)}{\sum_{k \in \ncal_0} dL_{ik}(t)} = 1$.
Consider now the case for $V_i(t) < 0$ and let $\tau = \sup\{s \leq t \; | \; V_i(s) = 0\}$.  Since $\sum_{j \in \ncal_0} a_{ij}(\tau) = 1$ by prior results, we will assume that $\sum_{j \in \ncal_0} a_{ij}(t) = 1$ to deduce
\begin{align*}
\sum_{j \in \ncal_0} da_{ij}(t) &= \sum_{j \in \ncal_0} \frac{dL_{ij}(t) - a_{ij}(t) \sum_{k \in \ncal_0} dL_{ik}(t)}{V_i(t)^-}\\
&= \frac{\sum_{j \in \ncal_0} dL_{ij}(t)}{V_i(t)^-} - \frac{\left(\sum_{j \in \ncal_0} a_{ij}(t)\right) \left(\sum_{k \in \ncal_0} dL_{ik}(t)\right)}{V_i(t)^-} = 0.
\end{align*}
Therefore based on the initial conditions, $a_{ij}(t)$ must evolve so that it maintains the constant row sum of $1$.
\end{enumerate}
\end{proof}

\subsection{Proof of Theorem~\ref{thm:continuous}}\label{sec:proof-continuous}
\begin{proof}
Recall that the initial values to the Eisenberg-Noe differential system are $V_i(0) > 0$ and $a_{ij}(0) = \frac{dL_{ij}(0)}{\sum_{k \in \ncal_0} dL_{ik}(0)}\ind{i \neq 0} + \frac{1}{n}\ind{i = 0,\; j \neq 0}$ for all banks $i,j \in \ncal_0$.
For ease of notation, consider $\tau_0 := 0$ and recursively define the stopping times
\[\tau_{m+1} := \inf\{t \in (\tau_m,T] \; | \; V_i(\tau_m)V_i(t) < 0 \text{ or } [V_i(\tau_m) = 0, \; dV_i(\tau_m)V_i(t) < 0]\}.\]
That is, $\tau_m \in \bbt$ is the time of the $m$th change in $\Lambda(V)$.  Without loss of generality, we will assume that $\tau_m = T$ if the infimum is taken over an empty set.  We note that the times $\tau_m$ are all stopping times with respect to the natural filtration.

With these times, note that in particular, on the interval $(\tau_m,\tau_{m+1}]$ we can consider the set of delinquent banks to be constant; to simplify, and slightly abuse, notation we can thus consider a constant matrix of delinquent firms $\Lambda(\tau_m)$ in the interval $(\tau_m,\tau_{m+1}]$.  We will now construct the unique solution forward in time over these time intervals, noting that we update $\Lambda$ and $\tau_{m+1}$ once the next event is found.

First, by construction, on $[0,\tau_1]$ there exists a unique solution to the differential system provided by $V(t) = V(0) + x(t) + L(t)^\T\vec{1} - L(t)\vec{1}$ and $a_{ij}(t) = \frac{dL_{ij}(t)}{\sum_{k \in \ncal_0} dL_{ik}(t)} \ind{i \neq 0} + \frac{1}{n}\ind{i = 0,\; j \neq 0}$ for all banks $i,j \in \ncal_0$.  Assume there exists a solution in the time interval $[0,\tau_m]$ for $\tau_m < T$.  Now we want to prove the existence and uniqueness for the clearing cash accounts and relative exposures on the interval $(\tau_m,\tau_{m+1}]$. 
Recall the construction of $dV(t)$:
\begin{align*}
    dV(t) = \left[ I -A(t)^\T \Gl(\tau_m) \right]^{-1}\left(dx(t) - [I - A(t)^\T]\dot{L}(t) dt \vec{1}\right).
\end{align*} 
Furthermore, assuming $A(t)$ exists, the network multiplier $\left[ I -A(t)^\T \Gl(\tau_m) \right]^{-1}$ is bounded:
\begin{align*}
\|(I - A(t)^\T\Lambda(\tau_m))^{-1}\|_1^{op} &\leq \sum_{k = 0}^\infty \|[A(t)^\T\Lambda(\tau_m)]^k\|_1^{op} \\
&\leq 1 + \sum_{k = 1}^\infty (1-\delta)^{k-1} = 1 + \frac{1}{\delta}.
\end{align*}
It therefore follows that, for a fixed $\Gl$, $dV(t)$ is locally Lipschitz in $A(t)$.

Now we wish to consider the differential form for the relative exposures matrix \eqref{eq:continuous-A}.  First, if $\Lambda_{ii}(\tau_m) = 0$ (and in particular, $\Lambda_{00}(\tau_m) = 0$ by assumption of the societal node) then $a_{ij}(t) = \frac{dL_{ij}(t)}{\sum_{k \in \ncal_0} dL_{ik}(t)}\ind{i \neq 0} + \frac{1}{n}\ind{i = 0, \; j \neq 0}$ is the unique solution for any firm $j \in \ncal_0$ over all times $t \in (\tau_m,\tau_{m+1}]$.  In particular, this is independent of the evolution of the cash accounts $V$, so we need only consider the joint differential equation between the cash accounts $V$ and the relative exposures $a_{ij}$ where bank $i$ is in delinquency between times $\tau_m$ and $\tau_{m+1}$, i.e., $\Lambda_{ii}(\tau_m) = 1$.  Consider bank $i \in \ncal$ with $\Lambda_{ii}(\tau_m) = 1$.  Therefore by construction $V_i(t) < 0$ for all $t \in (\tau_m,\tau_{m+1})$.  If $V_i(\tau_{m+1}) = 0$ then from Proposition~\ref{prop:A}, it already follows that the unique solution $a_{ij}(\tau_{m+1}) = \frac{dL_{ij}(\tau_{m+1})}{\sum_{k \in \ncal_0} dL_{ik}(\tau_{m+1})}$ must hold, otherwise we can extend $V_i(t) < 0$ for $t \in (\tau_m,\tau_{m+1}]$.  The differential form for all relative exposures \eqref{eq:continuous-A} on the interval $(\tau_m,\tau_{m+1}]$ is provided by $da_{ij}(t) = \frac{dL_{ij}(t) - a_{ij}(t)\sum_{k \in \ncal_0}dL_{ik}(t)}{V_i(t)^-}$.  By construction $(a_{ij},V_i) \in [0,1] \times -\bbr_{++} \mapsto \frac{\dot{L}_{ij}(t) - a_{ij} \sum_{k \in \ncal_0} \dot{L}_{ik}(t)}{-V_i}$ is locally Lipschitz. 

Combining our results for the joint differential system for the clearing cash accounts $V$ from \eqref{eq:continuous-V} and relative exposures $A$ from \eqref{eq:continuous-A}, we find that this system satisfies a joint  local Lipschitz property on the interval $(\tau_m,\tau_{m+1}]$.  Therefore, there exists some $\epsilon \in \lcal_{T}^{\infty}(\bbr_{++})$ (such that $\tau_m + \epsilon$ is a stopping time) for which a solution for $(V,A): [\tau_m,\tau_m + \epsilon] \to \bbr^{n+1} \times \bba$ exists and is unique.  Using the same logic with local properties, we can continue our unique solution sequentially.  This can be continued until the stopping time $\tau_{m+1}$ is reached (found along the path of $(V,A)$ as a stopping time) or this process reaches some maximal time $T^* < \tau_{m+1}$ for which a unique solution exists on the time interval $[\tau_m,T^*)$.  We note that any solution $V(t)$ must, almost surely, exist in the (almost surely) compact space
\begin{align*}
&\Big[V(\tau_m) - \left(I + \frac{1+\delta}{\delta} \mathbbm{1}\right)\left(\int_{\tau_m}^t dx(s)^- + \left(L(t) - L(\tau_m)\right)\vec{1}\right) , \\
&\quad V(\tau_m) + x(t)-x(\tau_m) + (L(t)-L(\tau_m))^\T\vec{1} - (L(t)-L(\tau_m))\vec{1}\Big]
\subseteq \lcal_t^2(\bbr^{n+1})
\end{align*}
where $\mathbbm{1} = \{1\}^{(n+1) \times (n+1)}$.  The lower bound is determined to be based on the bounding of the Leontief inverse; the upper bound follows from the continuous-time version of \eqref{eq:discrete-Vdt-expand}, i.e.,
\[V(t) = V(0) + x(t) + L(t)^\T\vec{1} - L(t)\vec{1} - A(t)^\T V(t)^-.\]
Additionally, $a_{ij}(t)$ almost surely exists in the compact neighborhood $[0,1]$ by definition.
Therefore $(V(T^*),A(T^*)) = \lim_{t \nearrow T^*} (V(t),A(t))$ exists by continuity of the solutions and compactness of the range space.
Thus we can continue the differential equation from time $T^*$ with values $(x(T^*),V(T^*),A(T^*))$ which contradicts the nature that $T^*$ is the maximal time.  Notably, if $V_i(T^*) = 0$ for some bank $i$ then it is imperative to check if $\tau_{m+1} = T^*$ to update the set of delinquent banks $\Lambda$.

Therefore, by induction, there exists a unique solution $(V,A)$ to \eqref{eq:continuous-V} and \eqref{eq:continuous-A} on the domain $[0,\tau_m]$ for any index $m \in \bbn$.  In particular this holds up to $\tau^* = \sup_{m \in \bbn} \tau_m$.  
If $\tau^* \geq T$ then the proof is complete.  If $\tau^* < T$, then by the same argument as above we can find $(V(\tau^*),A(\tau^*))$ as we can bound both the cash accounts and relative exposures into an almost surely compact neighborhood (and a subset of $\lcal_{\tau^*}^2(\bbr^{n+1})$).  Therefore, as before, we can start the process again at time $\tau^*$, which contradicts the terminal nature of $\tau^*$.  This concludes the proof.

\end{proof}

\subsection{Proof of Proposition~\ref{prop:cash}}\label{sec:proof-cash}
\begin{proof}
This clearing problem is a contraction mapping due to the continuous piecewise linearity of the clearing problem along with the obligations to society; therefore uniqueness follows from the Banach fixed point theorem.
\end{proof}

\subsection{Proof of Corollary~\ref{cor:defaults}}\label{sec:proof-defaults}
\begin{proof}
This result follows by an application of Theorem~\ref{thm:continuous}.  Specifically, we need only verify that there exists a clearing solution at each default time as the inter-default times follow exactly the structure of the no-default scenario presented in the prior section (with the addition of the capital account $K$).
Consider the default time $\tau_j$ of bank $j\in\ncal$.  We will construct a greatest clearing solution following an argument similar to the fictitious default argument of~\cite{EN01,RV13}.
First, assume there are no simultaneous defaults; by Proposition~\ref{prop:cash}, there exists a greatest joint cash account and payments at that time. These payments, however, impact the capital account of all non-defaulting firms. If no defaults are triggered then this is a greatest clearing solution at $\tau_j$ and we can continue to the next default time.  If any additional defaults are triggered then the clearing problem presented in Proposition~\ref{prop:cash} is rerun with the inclusion of these simultaneous defaults.  This procedure is continued until no new defaults occur and it results in the smallest possible default cascade at $\tau_j$.
\end{proof}

\section{Derivation of Inter-Default Time Continuous-Time Clearing from Discrete-Time Clearing}\label{sec:discrete-top}
\subsection{Inter-Default Time Discrete-Time Clearing}\label{sec:discrete}
Consider now a discrete set of clearing times $\bbt$, e.g., $\bbt = \{0,1,\dots,T\}$ for some (finite) terminal time $T < \infty$ or $\bbt = \bbn$.
Such a setting is presented in \cite{CC15}.
As before, we will use the notation from \cite{cont2013ito} such that the process $Z: \bbt \to \bbr^n$ has value of $Z(t)$ at time $t \in \bbt$ and history $Z_t := (Z(s))_{s = 0}^t$.

\begin{figure}
\centering
\begin{subfigure}[t]{\linewidth}
\centering
\begin{tikzpicture}[x=\linewidth/10]
\draw[draw=none] (0,9.5) rectangle (6,10) node[pos=.5]{\bf Balance Sheet};
\draw[draw=none] (0,9) rectangle (3,9.5) node[pos=.5]{\bf Assets};
\draw[draw=none] (3,9) rectangle (6,9.5) node[pos=.5]{\bf Liabilities};

\filldraw[fill=blue!20!white,draw=none] (0,7) rectangle (3,9) node[pos=.5,style={align=center}]{Cash-Flow @ $t = 0$ \\ $\Delta  x_i(0)$};
\filldraw[fill=blue!20!white,draw=none] (0,6) rectangle (3,7) node[pos=.5,style={align=center}]{Cash-Flow @ $t = 1$ \\ $\Delta  x_i(1)$};
\filldraw[fill=yellow!20!white,draw=none] (0,2) rectangle (3,6) node[pos=.5,style={align=center}]{Interbank @ $t = 0$ \\ $\sum_{j = 1}^n \pi_{ji}(0)p_j(0)$};
\filldraw[fill=yellow!20!white,draw=none] (0,0) rectangle (3,2) node[pos=.5,style={align=center}]{Interbank @ $t = 1$ \\ $\sum_{j = 1}^n \pi_{ji}(1)p_j(1)$};

\filldraw[fill=purple!20!white,draw=none] (3,5) rectangle (6,9) node[pos=.5,style={align=center}]{Cash-Flow @ $t = 0$ \\ $\sum_{j = 1}^n \Delta  L_{ij}(0)$};
\filldraw[fill=purple!20!white,draw=none] (3,1.5) rectangle (6,5) node[pos=.5,style={align=center}]{Cash-Flow @ $t = 1$ \\ $\sum_{j = 1}^n \Delta  L_{ij}(1)$};
\filldraw[fill=orange!20!white,draw=black] (3,0) rectangle (6,1.5) node[pos=.5,style={align=center}]{Cash Account \\ $V_i(1)$};

\draw[draw=black] (0,0) rectangle (3,9);
\draw[draw=black] (3,0) rectangle (6,9);

\draw[-,dashed] (0,7) -- (3,7);
\draw[-] (0,6) -- (3,6);
\draw[-,dashed] (0,2) -- (3,2);
\draw[-,dashed] (3,5) -- (6,5);
\end{tikzpicture}
\caption{Stylized actualized balance sheet for firm $i$ with two time periods.}
\label{fig:BS}
\end{subfigure}
~\\[5mm]
\begin{subfigure}[t]{\linewidth}
\begin{tikzpicture}[x=\linewidth/17]
\draw[draw=none] (0,6.5) rectangle (6,7) node[pos=.5]{\bf Cash Flow Statement @ $t = 0$};
\draw[draw=none] (0,6) rectangle (3,6.5) node[pos=.5]{\bf Assets};
\draw[draw=none] (3,6) rectangle (6,6.5) node[pos=.5]{\bf Liabilities};

\filldraw[fill=blue!20!white,draw=black] (0,4) rectangle (3,6) node[pos=.5,style={align=center}]{Cash-Flow \\ $\Delta  x_i(0)$};
\filldraw[fill=yellow!20!white,draw=black] (0,0) rectangle (3,4) node[pos=.5,style={align=center}]{Interbank \\ $\sum_{j = 1}^n \pi_{ji}(0)p_j(0)$};

\filldraw[fill=purple!20!white,draw=black] (3,2) rectangle (6,6) node[pos=.5,style={align=center}]{Cash-Flow \\ $\sum_{j = 1}^n \Delta  L_{ij}(0)$};
\filldraw[fill=orange!20!white,draw=black] (3,0) rectangle (6,2) node[pos=.5,style={align=center}] (t) {Cash Account \\ $V_i(0)$};

\draw[draw=none] (7,6.5) rectangle (13,7) node[pos=.5]{\bf Cash Flow Statement @ $t = 1$};
\draw[draw=none] (7,6) rectangle (10,6.5) node[pos=.5]{\bf Assets};
\draw[draw=none] (10,6) rectangle (13,6.5) node[pos=.5]{\bf Liabilities};

\foreach \y in {0.5}
{
\filldraw[fill=blue!20!white,draw=black] (7,4+\y) rectangle (10,5+\y) node[pos=.5,style={align=center}]{Cash-Flow \\ $\Delta  x_i(1)$};
\filldraw[fill=green!20!white,draw=black] (7,2+\y) rectangle (10,4+\y) node[pos=.5,style={align=center}] (t+1) {Carry-Forward \\ $V_i(0)^+$};
\filldraw[fill=yellow!20!white,draw=black] (7,0+\y) rectangle (10,2+\y) node[pos=.5,style={align=center}]{Interbank \\ $\sum_{j = 1}^n \pi_{ji}(1)p_j(1)$};

\filldraw[fill=purple!20!white,draw=black] (10,1.5+\y) rectangle (13,5+\y) node[pos=.5,style={align=center}]{Cash-Flow \\ $\sum_{j = 1}^n \Delta  L_{ij}(1)$};
\filldraw[fill=red!20!white,draw=none] (13.2,1+\y) rectangle (16,2+\y) node[pos=.5,style={align=center}]{Carry-Forward \\ $V_i(0)^- = 0$};
\filldraw[fill=red!20!white,draw=none] (13,1.5+\y) -- (13.2,2+\y) -- (13.2,1+\y) -- cycle;
\draw (13,1.5+\y) -- (13.2,2+\y) -- (16,2+\y) -- (16,1+\y) -- (13.2,1+\y) -- cycle;
\filldraw[fill=orange!20!white,draw=black] (10,0+\y) rectangle (13,1.5+\y) node[pos=.5,style={align=center}]{Cash Account \\ $V_i(1)$};

\draw[->,line width=1mm] (6,1) -- (7,3+\y);
\draw[->,dotted,line width=1mm] (6,.25) -- (14,.25) -- (14,1+\y);
}
\end{tikzpicture}
\caption{Stylized cash flow statement of actualized balance sheet for firm $i$ at times $0$ and $1$.}
\label{fig:discrete-BS}
\end{subfigure}
\caption{Comparison of the full balance sheet to the cash flow statements utilized for Section~\ref{sec:discrete}.}
\label{fig:BalanceSheet}
\end{figure}

In this setting, we will consider the external (incoming) cash flow $\Delta  x: \bbt \to \bbr^{n+1}_+$ and nominal liabilities $\Delta  L: \bbt \to \bbr^{(n+1) \times (n+1)}_+$ to be functions of the clearing time, i.e., as assets and liabilities with different maturities. The external cash in-flows and nominal liabilities can explicitly depend on the clearing results of the prior times (i.e., $\Delta  x(t,V_{t-1})$ and $\Delta  L(t,V_{t-1})$) without affecting the existence and uniqueness results we present, but for simplicity of notation we will focus on the case where the external assets and nominal liabilities are independent of the health and wealth of the firms.
As introduced in the mean body of the text, throughout we are considering the discounted assets and liabilities so as to simplify notation.

In contrast to the static Eisenberg-Noe framework, herein we need to consider the results of the prior times.  In particular, if firm $i$ has positive equity at time $t-1$ (i.e., $V_i(t-1) > 0$) then these additional assets are available to firm $i$ at time $t$ in order to satisfy its obligations.  Similarly, if firm $i$ has negative cash account at time $t-1$  (i.e., $V_i(t-1) < 0$) then the debts that the firm has not yet paid will roll-forward in time and be due at the next period.  

In this way a firm can be considered in \emph{delinquency} at a time if it is unable to satisfy its obligations at that time. Defaults are treated in continuous time within the main body of this work. See Figure~\ref{fig:discrete-BS} for a stylized cash flow statement example for a firm that has positive cash account at time $0$ that rolls forward to time $1$.  The full (actualized) balance sheet for this example with only those two time periods is displayed in Figure~\ref{fig:BS}; we note that the full balance sheet as depicted considers actualized payments rather than the book value of the obligations.

\begin{assumption}\label{ass:initial}
Before the time of interest, all firms are solvent and liquid.  That is, $V_i(-1) \geq 0$ for all firms $i \in \ncal_0$.
\end{assumption}

We can now construct the total liabilities and relative liabilities at time $t \in \bbt$ as
\begin{align*}
\bar p_i(t,V_{t-1}) &:= \sum_{j \in \ncal_0} \Delta  L_{ij}(t) + V_i(t-1)^-\\
\pi_{ij}(t,V_{t-1}) &:= \begin{cases} \frac{\Delta  L_{ij}(t) + \pi_{ij}(t-1,V_{t-2})V_i(t-1)^-}{\bar p_i(t,V_{t-1})} &\text{if } \bar p_i(t,V_{t-1}) > 0 \\ \frac{1}{n} &\text{if } \bar p_i(t,V_{t-1}) = 0, \; j \neq i\\ 0 &\text{if } \bar p_i(t,V_{t-1}) = 0, \; j = i \end{cases} \quad \forall i,j \in \ncal_0.
\end{align*}
In this way, coupled with the accumulation of positive cash account over time, the clearing cash accounts must satisfy the following fixed point problem in time $t$ cash accounts:
\begin{equation}\label{eq:EN-discrete}
V(t) = V(t-1)^+ + \Delta  x(t) + \Pi(t,V_{t-1})^\T \left[\bar p(t,V_{t-1}) - V(t)^-\right]^+ - \bar p(t,V_{t-1}).
\end{equation}
That is, all firms have a clearing cash account that is the summation of their positive equity at the prior time, the new incoming external cash flow, and the payments made by all other firms minus the total obligations of the firm (including the prior unpaid liabilities).  In this way we can construct the cash accounts of firms forward in time. This can be considered a discrete-time extension of \eqref{eq:EN-e}.

We now wish to consider a reformulation of \eqref{eq:EN-discrete} with a functional relative exposures matrix $A$, i.e.,
\begin{align}
\label{eq:discrete-A}
a_{ij}(t,V_t) &:= \begin{cases} \pi_{ij}(t,V_{t-1}) &\text{if } \bar p_i(t,V_{t-1}) \geq V_i(t)^-\\ \frac{\Delta  L_{ij}(t) + \pi_{ij}(t-1,V_{t-2})V_i(t-1)^-}{V_i(t)^-} &\text{if } \bar p_i(t,V_{t-1}) < V_i(t)^-\end{cases} \quad \forall i,j \in \ncal_0.
\end{align}
Here we introduce the functional matrix $A: \bbt \times \bbr^{n+1} \to [0,1]^{(n+1) \times (n+1)}$ to be the relative exposure matrix.  That is, $a_{ij}(t,V_t)V_i(t)^-$ provides the (negative) impact that firm $i$'s losses have on firm $j$'s cash account at time $t \in \bbt$.  This is in contrast to $\Pi$, the relative liabilities, in that it endogenously imposes the limited exposures concept.  In this work the two notions will generally coincide, but for mathematical simplicity we introduce this relative exposure matrix.  For the equivalence we seek, we define the relative exposures so that
\[\Delta  L(t)^\T\vec{1} + A(t-1,V_{t-1})^\T V(t-1)^- - A(t,V_t)^\T V(t)^- = \Pi(t,V_{t-1})^\T [\bar p(t,V_{t-1}) - V(t)^-]^+\]
for any $V(t) \in \bbr^{n+1}$.  This formulation is such that if the positive part were removed from the right hand side, the relative exposures $A$ would be defined exactly as the relative liabilities $\Pi$ by construction.
In particular, we will define the relative exposures element-wise and pointwise so as to encompass the limited exposures as in \eqref{eq:discrete-A}.  If $\bar p_i(t,V_{t-1}) > 0$ then we can simplify this further as $a_{ij}(t,V_t) = \frac{\Delta  L_{ij}(t) + a_{ij}(t-1,V_{t-1})V_i(t-1)^-}{\max\{\bar p_i(t,V_{t-1}) , V_i(t)^-\}}$.

Using the notation and terms above we can rewrite \eqref{eq:EN-discrete} with respect to the cash flows $\Delta  x$ and relative exposures $A$ as
\begin{align}
\nonumber V(t) &= V(t-1)^+ + \Delta  x(t) + \Pi(t,V_{t-1})^\T [\bar p(t,V_{t-1}) - V(t)^-]^+ - \bar p(t,V_{t-1})\\
\nonumber &= V(t-1)^+ + \Delta  x(t) + \Delta  L(t)^\T\vec{1} + A(t-1,V_{t-1})^\T V(t-1)^-\\
\nonumber &\qquad - A(t,V_t)^\T V(t)^- - \Delta  L(t)\vec{1} -V(t-1)^-\\
\label{eq:discrete-V} &= V(t-1) + \Delta  x(t) + \Delta  L(t)^\T\vec{1} + A(t-1,V_{t-1})^\T V(t-1)^- - A(t,V_t)^\T V(t)^- - \Delta  L(t)\vec{1}.
\end{align}

With this setup we now wish to extend the existence and uniqueness results of \cite{EN01} to discrete-time.
\begin{theorem}\label{thm:discrete}
Let $(\Delta x,\Delta L): \bbt \to \bbr^{n+1}_+ \times \bbr^{(n+1) \times (n+1)}_+$ define a dynamic financial network such that every bank owes to the societal node at all times $t \in \bbt$, i.e., $L_{i0}(t) > 0$ for all banks $i \in \ncal$ and times $t \in \bbt$.  Under Assumption~\ref{ass:initial}, there exists a unique solution of clearing cash accounts $V: \bbt \to \bbr^{n+1}$ to \eqref{eq:discrete-V}.
\end{theorem}
\begin{proof}
We will prove this result inductively.  First consider time $t = 0$. Recall from Assumption~\ref{ass:initial} that $V(-1) \geq 0$. The clearing cash accounts at time $0$ follow the fixed point equation
\[V(0) = \Phi(0,V(0)) := V(-1) + \Delta x(0) + \Delta L(0)^\T\vec{1} - A(0,V_0)^\T V(0)^- - \Delta L(0)\vec{1}.\]
Note that, by construction, $A(0,V_0)^\T V(0)^- \leq \Delta L(0)^\T \vec{1}$.  Therefore any clearing solution must fall within the compact range $[V(-1) + \Delta x(0) -\Delta L(0) \vec{1},V(-1) + \Delta x(0)+\Delta L(0)^\T\vec{1}-\Delta L(0)\vec{1}] \subseteq \bbr^{n+1}$.  It is clear from the definition that $\Phi(0,\cdot)$ is a monotonic operator, and thus there exists a greatest and least clearing solution $V^\uparrow(0) \geq V^\downarrow(0)$ by Tarski's fixed point theorem \cite[Theorem 11.E]{Z86}, both of which must fall within this domain.  Further, $a_{ij}(0,V_0) = \frac{\Delta L_{ij}(0)}{\sum_{k \in \ncal_0}\Delta L_{ik}(0)}$ (for $i \in \ncal$ and $j \in \ncal_0$) for any cash account $V(0)$ in this domain since $V(-1) + \Delta x(0) - \Delta L(0) \vec{1} \geq -\Delta L(0) \vec{1} = -\bar p(0,V_{-1})$.  We will prove uniqueness as it is done in \cite{EN01} by noting additionally that we can assume that the societal node will always have positive equity (i.e., $V^\downarrow(0) \geq 0$).  First, we will show that the positive equities are the same for every firm no matter which clearing solution is chosen, i.e., $V_i^\uparrow(0)^+ = V_i^\downarrow(0)^+$ for every firm $i \in \ncal_0$.  By definition $V^\uparrow(0) \geq V^\downarrow(0)$ and using $\sum_{j \in \ncal_0} a_{ij}(0) = 1$ for every firm $i \in \ncal_0$ we recover
\begin{align*}
\sum_{i \in \ncal_0} V_i^\uparrow(0)^+ &= \sum_{i \in \ncal_0} \left[V_i^\uparrow(0) + V_i^\uparrow(0)^-\right]\\
&= \sum_{i \in \ncal_0} \left[V_i(-1) + \Delta x_i(0) + \sum_{j \in \ncal} \Delta L_{ji}(0) - \sum_{j \in \ncal} a_{ji}(0,V_0^\uparrow) V_j^\uparrow(0)^- - \sum_{j \in \ncal_0} \Delta L_{ij}(0) + V_i^\uparrow(0)^-\right]\\
&= \sum_{i \in \ncal_0} \left[V_i(-1) + \Delta x_i(0) + \sum_{j \in \ncal} \Delta L_{ji}(0) - \sum_{j \in \ncal_0} \Delta L_{ij}(0)\right]\\
&\qquad - \sum_{j \in \ncal} V_j^\uparrow(0)^- \sum_{i \in \ncal_0} a_{ji}(0,V_0^\uparrow) + \sum_{i \in \ncal_0} V_i^\uparrow(0)^-\\
&= \sum_{i \in \ncal_0} \left[V_i(-1) + \Delta x_i(0) + \sum_{j \in \ncal} \Delta L_{ji}(0) - \sum_{j \in \ncal_0} \Delta L_{ij}(0)\right] = \sum_{i \in \ncal_0} V_i^\downarrow(0)^+.
\end{align*}
Therefore it must be the case that $V_i^\uparrow(0)^+ = V_i^\downarrow(0)^+$ for all firms $i \in \ncal_0$.
Since we assume that the societal node will always have positive equity, it must be the case that $V_0^\uparrow(0) = V_0^\downarrow(0)$.  Now since we assume that each node $i \in \ncal$ owes to the societal node, if any firm $i \in \ncal$ is such that $0 \geq V_i^\uparrow(0) > V_i^\downarrow(0)$ then it must be that $V_0^\uparrow(0) > V_0^\downarrow(0)$, which is a contraction.

Continuing with the inductive argument, assume that the history of clearing cash accounts $V_{t-1}$ up to time $t-1$ is fixed and known.  The clearing cash accounts at time $t$ follow the fixed point equation
\[V(t) = \Phi(t,V(t)) := V(t-1) + \Delta x(t) + \Delta L(t)^\T\vec{1} - A(t,V_t)^\T V(t)^- + A(t-1,V_{t-1})^\T V(t-1)^- - \Delta L(t)\vec{1}.\]
Note that, by construction, $A(t,V_t)^\T V(t)^- \leq \Delta L(t)^\T \vec{1} + A(t-1,V_{t-1})^\T V(t-1)^-$.  Therefore any clearing solution must fall within the compact range $[V(t-1) + \Delta x(t) - \Delta L(t)\vec{1} , V(t-1) + \Delta x(t) + \Delta L(t)^\T\vec{1} + A(t-1,V_{t-1})^\T V(t-1)^- - \Delta L(t)\vec{1}] \subseteq \bbr^{n+1}$.  Further, $a_{ij}(t,V_t) = \frac{\Delta L_{ij}(t) + a_{ij}(t-1,V_{t-1})V_i(t-1)^-}{\sum_{k \in \ncal_0} \Delta L_{ik}(t) + V_i(t-1)^-}$ (for $i \in \ncal$ and $j \in \ncal_0$) for any cash accounts $V(t)$ in this domain since $V(t-1) + \Delta x(t) - \Delta L(t)\vec{1} \geq -V(t-1)^- - \Delta L(t)\vec{1} = -\bar p(t,V_{t-1})$.  Thus we can apply the same logic as in the time $0$ case to recover existence and uniqueness of the clearing cash accounts $V(t)$ at time $t$.
\end{proof}

\begin{remark}\label{rem:regularnetwork}
The assumption that all firms have obligations to the societal node $0$ at all times $t \in \bbt$ guarantees that the financial system is a ``regular network'' (see \cite[Definition 5]{EN01}) at all times.
\end{remark}

\subsection{Discrete to Continuous-Time}\label{sec:disc-to-cont}
In the prior section on a discrete-time model for clearing cash accounts, we implicitly assumed a constant time-step between each clearing date of $\Dt = 1$ throughout.  In order to construct a continuous-time clearing model we will begin by making a discrete-time model with an explicit $\Dt>0$ term.  In fact, this is immediate from the prior construction with a minor alteration to the cash flow term.  Herein we construct the external cash flow at time $t$ to be given by $\Gd x(t,\Dt) := \int_{t-\Dt}^t dx(s)$ and the nominal liabilities at time $t$ are similarly provided by $\Gd L(t,\Dt) := \int_{t-\Dt}^t dL(s)$ where both $dx$ and $dL$ are discussed in Section~\ref{sec:continuous} (additionally, we set $dx(t) = 0$ and $dL(t) = 0$ for any times $t < 0$).  

With these parameters we can construct the $\Dt$-discrete-time clearing process $V(t,\Dt)$ and exposure matrix $A(t,\Dt,V_t(\Dt))$ by:
\begin{align}
\label{eq:discrete-Vdt} \begin{split}V(t,\Dt) &= V(t-\Dt,\Dt) + \Gd x(t,\Dt) + \Gd L(t,\Dt)^\T\vec{1} - A(t,\Dt,V_t(\Dt))^\T V(t,\Dt)^-\\ &\qquad + A(t-\Dt,\Dt,V_{t-\Dt}(\Dt))^\T V(t-\Dt,\Dt)^- - \Gd L(t,\Dt)\vec{1}\end{split}\\
\label{eq:discrete-Adt} \begin{split}a_{ij}(t,\Dt,V_t(\Dt)) &= \frac{\Gd L_{ij}(t,\Dt)+a_{ij}(t-\Dt,\Dt,V_{t-\Dt}(\Dt))V_i(t-\Dt,\Dt)^-}{\max\{\sum_{k \in \ncal_0} \Gd L_{ik}(t,\Dt)+V_i(t-\Dt,\Dt)^- , V_i(t,\Dt)^-\}}\ind{i \neq 0}\\ &\qquad + \frac{1}{n}\ind{i = 0,\; j \neq 0} \qquad \forall i,j \in \ncal_0.\end{split}
\end{align}
Here we assume that $V(t) = V(-1) \geq 0$ for every time $t < 0$ as in Assumption~\ref{ass:initial}.  This construction can be computed either in continuous-time $t \in \bbt$ with sliding intervals of size $\Dt$ or at the discrete times $t \in \{0,\Dt,...,T\}$.
The existence and uniqueness of this system follow exactly as in Theorem~\ref{thm:discrete} under Assumption~\ref{ass:society}.

\begin{corollary}\label{cor:discrete-cont}
Let $(dx,dL): \bbt \to \bbr^{n+1}_+ \times \bbr^{(n+1) \times (n+1)}_+$ define a dynamic financial network satisfying Assumption~\ref{ass:society}.  Under Assumption~\ref{ass:initial}, there exists a unique solution of clearing cash accounts $V: \bbt \times \bbr_{++} \to \bbr^{n+1}$ to \eqref{eq:discrete-Vdt}.  Further, the clearing cash accounts are jointly continuous in time and step-size.
\end{corollary}
\begin{proof}
Existence and uniqueness of the clearing solutions follows from Theorem~\ref{thm:discrete}.  To prove continuity we will employ an induction argument.  To do so, we will consider the reduced domain $V: \bbt \times [\epsilon,\infty) \to \bbr^{n+1}$ for some $\epsilon > 0$.  That is, we restrict the step-size $\Dt \geq \epsilon$.  As we will demonstrate that the continuity argument holds for any $\epsilon > 0$ then the desired result must hold as well.  Before continuing, consider an expanded version of the recursive formulation of \eqref{eq:discrete-Vdt}, i.e.,
\begin{equation}\label{eq:discrete-Vdt-expand}
V(t,\Dt) = V(-1) + \int_0^t dx(s) + \int_0^t dL(s)^\T\vec{1} - A(t,\Dt,V_t(\Dt))^\T V(t,\Dt)^- - \int_0^t dL(s)\vec{1}
\end{equation}
for all times $t \in \bbt$.  Fix the minimal step-size $\epsilon > 0$.  Note that the relative exposures satisfy $a_{ij}(t,\Dt,V_t(\Dt)) := \frac{\int_0^t dL_{ij}(s)}{\sum_{k \in \ncal_0} \int_0^t dL_{ik}(s)}$ for any time $t \in [0,\epsilon)$ by the assumption that $V(-1) \geq 0$.  Thus we can conclude $V: [0,\epsilon) \times [\epsilon,\infty) \to \bbr^{n+1}$ is continuous by an application of \cite[Proposition A.2]{feinstein2014measures}.
Now, by way of induction, assume that $V: [0,s) \times [\epsilon,\infty) \to \bbr^{n+1}$ is continuous for some $s > 0$.  Again, by \cite[Proposition A.2]{feinstein2014measures}, we are able to immediately conclude that $V: [0,s+\epsilon) \cap \bbt \times [\epsilon,\infty) \to \bbr^{n+1}$ is continuous.  As we are able to always extend the continuity result by $\epsilon > 0$ in time, the result is proven.
\end{proof}

Now we want to consider the limiting behavior of this discrete-time system as $\Dt$ tends to 0.
To do so, first, we will consider the formulation of the relative exposures $a_{ij}$ from bank $i$ to $j$.  From Corollary~\ref{cor:discrete-cont} and Assumption~\ref{ass:society}, we know that for any time $t \in \bbt$ and bank $i \in \ncal$ it must follow that $\sum_{k \in \ncal_0} \Gd L_{ik}(t,\Dt) + V_i(t-\Dt,\Dt)^- \geq V_i(t,\Dt)^-$ for $\Dt > 0$ small enough due to the joint continuity of the cash accounts in time and step-size.  Thus in the limiting case, as $\Dt \searrow 0$, we find that we can consider the relative liabilities rather than the relative exposures, i.e., for $\Dt$ small enough
\begin{equation}
\label{eq:discrete-Adt-limiting}
\begin{split}a_{ij}(t,\Dt,V_t(\Dt)) &= \frac{\Gd L_{ij}(t,\Dt)+a_{ij}(t-\Dt,\Dt,V_{t-\Dt}(\Dt))V_i(t-\Dt,\Dt)^-}{\sum_{k \in \ncal_0} \Gd L_{ik}(t,\Dt)+V_i(t-\Dt,\Dt)^-}\ind{i \neq 0}\\
&\qquad + \frac{1}{n}\ind{i = 0,\; j \neq 0} \qquad \forall i,j \in \ncal_0.\end{split}
\end{equation}
Rearranging these terms we are able to deduce that, for any firm $i \in \ncal$,
\begin{equation}\label{eq:discrete-Adt-limiting2}
\begin{split}&[a_{ij}(t,\Dt,V_t(\Dt)) - a_{ij}(t-\Dt,\Dt,V_{t-\Dt}(\Dt))] V_i(t-\Dt,\Dt)^- \\ &\qquad = \Gd L_{ij}(t,\Dt) - a_{ij}(t,\Dt,V_t(\Dt)) \sum_{k \in \ncal_0} \Gd L_{ik}(t,\Dt).\end{split}
\end{equation}

Coupled with the assumption that the societal node always has positive cash account, we are thus able to consider the limiting behavior of \eqref{eq:discrete-Vdt} as the step-size $\Dt$ tends to 0.  To do so, consider
\begin{align}
\begin{split}V(t,\Dt) &= V(t-\Dt,\Dt) + \Gd x(t,\Dt) + \Gd L(t,\Dt)^\T\vec{1} - A(t,\Dt,V_t(\Dt))^\T V(t,\Dt)^-\\ &\qquad + A(t-\Dt,\Dt,V_{t-\Dt})^\T V(t-\Dt,\Dt)^- - \Gd L(t,\Dt)\vec{1}\end{split}\\
\begin{split} &= V(t-\Dt,\Dt) + \Gd x(t,\Dt) + \Gd L(t,\Dt)^\T\vec{1} - \Gd L(t,\Dt)\vec{1}\\ &\qquad - A(t,\Dt,V_t(\Dt))^\T V(t,\Dt)^- + A(t,\Dt,V_t(\Dt))^\T V(t-\Dt,\Dt)^-\\ &\qquad - A(t,\Dt,V_t(\Dt))^\T V(t-\Dt,\Dt)^- + A(t-\Dt,\Dt,V_{t-\Dt})^\T V(t-\Dt,\Dt)^-\end{split}\\
\begin{split} &= V(t-\Dt,\Dt) + \Gd x(t,\Dt) - A(t,\Dt,V_t(\Dt))^\T [V(t,\Dt)^- - V(t-\Dt,\Dt)^-]\\ &\qquad - [I - A(t,\Dt,V_t(\Dt))^\T] \Gd L(t,\Dt) \vec{1}. \label{eq:f_limit}\end{split}
\end{align}
Consider the matrix of delinquent firms, $\Lambda(V) \in \{0,1\}^{(n+1) \times (n+1)}$ to be the diagonal matrix of banks in delinquency, i.e.
\[\Lambda_{ij}(V) = \begin{cases}1 &\text{if } i = j \neq 0 \text{ and } V_i < 0 \\ 0 &\text{else}\end{cases} \quad \forall i,j \in \ncal_0.\]
We are able to set $\Lambda_{00}(V) = 0$ without loss of generality since, by assumption, the outside node $0$ has no obligations into the system.
Thus, as with \eqref{eq:discrete-Adt-limiting}, by continuity of the clearing cash accounts and $\Dt$ small enough, we can conclude that except at specific event times it follows that $\Lambda(V(t,\Dt)) = \Lambda(V(t-\Dt,\Dt))$.  Thus, with this added notation we can reformulate the clearing cash accounts equation \eqref{eq:discrete-Vdt} as
\begin{align*}
\begin{split}V(t,\Dt) &= V(t-\Dt,\Dt) + A(t,\Dt,V_t(\Dt))^\T \Lambda(V(t,\Dt)) [V(t,\Dt) - V(t-\Dt,\Dt)] + \Gd x(t,\Dt) \\ &\qquad - [I - A(t,\Dt,V_t(\Dt))^\T] \Gd L(t,\Dt) \vec{1}.\end{split}
\end{align*}
For the construction of a differential form we can consider the equivalent formulation
\begin{align}
\label{eq:discrete-Vdt-limiting}
\begin{split} V(t,\Dt) &- V(t-\Dt,\Dt) =\\
&[I - A(t,\Dt,V_t(\Dt))^\T \Lambda(V(t,\Dt))]^{-1} \left(\begin{array}{l}\Gd x(t,\Dt) \\ \quad - [I - A(t,\Dt,V_t(\Dt))^\T] \Gd L(t,\Dt) \vec{1}\end{array}\right).\end{split}
\end{align}
Note that $I - A(t,\Dt,V_t(\Dt))^\T \Lambda(V(t,\Dt))$ is invertible by standard input-output results and as proven in Proposition~\ref{prop:Leontief}.

Utilizing \eqref{eq:discrete-Vdt-limiting} and \eqref{eq:discrete-Adt-limiting} and taking the limit as $\Dt \searrow 0$, we are thus able to construct the joint differential system of \eqref{eq:continuous-V} and \eqref{eq:continuous-A}, i.e.,
\begin{align*}
dV(t) &= [I - A(t)^\T \Lambda(V(t))]^{-1} \left(dx(t) - [I - A(t)^\T] dL(t) \vec{1}\right)\\
da_{ij}(t) &= \begin{cases} \frac{d^2L_{ij}(t) - a_{ij}(t)\sum_{k \in \ncal_0} d^2L_{ik}(t)}{\sum_{k \in \ncal_0} dL_{ik}(t)} &\text{if } i \in \ncal, \; V_i(t) \geq 0 \\ \frac{dL_{ij}(t) - a_{ij}(t)\sum_{k \in \ncal_0} dL_{ik}(t)}{V_i(t)^-} &\text{if } i \in \ncal, \; V_i(t) < 0 \\ 0 &\text{if } i = 0\end{cases} \quad \forall i,j \in \ncal_0
\end{align*}
with initial conditions $V(0) \geq 0$ given and $a_{ij}(0) = \frac{dL_{ij}(0)}{\sum_{k \in \ncal_0} dL_{ik}(0)}\ind{i \neq 0} + \frac{1}{n}\ind{i = 0,\; j \neq 0}$ for all firms $i,j \in \ncal_0$.
As in \eqref{eq:discrete-Vdt-limiting}, $I - A(t)^\T \Lambda(V(t))$ is invertible by standard input-output results and as proven in Proposition~\ref{prop:Leontief}.
The first case in \eqref{eq:continuous-A} is constructed by noting that $a_{ij}(t) = \frac{dL_{ij}(t)}{\sum_{k \in \ncal_0} dL_{ik}(t)}$ if $V_i(t) \geq 0$ and $i \in \ncal$ and $da_{0j}(t) = 0$ for any firm $j \in \ncal_0$ for all times $t$; the second case in \eqref{eq:continuous-A} follows from \eqref{eq:discrete-Adt-limiting2} and taking the limit as $\Dt \searrow 0$.  Note that this differential system is discontinuous, with events at times when firms cross the 0 cash boundary, i.e., when $\Lambda(V(t)) \neq \Lambda(V(t^-))$.  As such, we will consider the differential system on the inter-event intervals, then update the differential system between these intervals.  This is made more explicit in the proof of Theorem~\ref{thm:continuous}.
As with the discrete-time system~\eqref{eq:discrete-Adt}, the relative exposures follow the incoming proportional obligations if a firm has a surplus of cash.  When a firm is in delinquency, the relative exposures follow a path that provides the average relative obligations between new liabilities and the prior unpaid liabilities.

\section{The Static Eisenberg-Noe Model as a Differential System}\label{sec:discussion-EN}

Herein we will consider the case in which the relative liabilities are constant through time.  That is, we consider the setting in which $dL_{ij}(s)/\sum_{k \in \ncal_0} dL_{ik}(s) = dL_{ij}(t)/\sum_{k \in \ncal_0} dL_{ik}(t)$ for all times $s,t \in \bbt$ and firms $i,j \in \ncal_0$ so long as $\sum_{k \in \ncal_0} dL_{ik}(s), \sum_{k \in \ncal_0} dL_{ik}(t) > 0$.  Following Assumption~\ref{ass:society}, the gross marginal liabilities of all banks $i \in \ncal$ and at all times $t \in \bbt$ are strictly great than 0. The key implication of this assumption is that the relative exposures matrix in~\eqref{eq:continuous-A} can be found explicitly to equal the relative liabilities
\[a_{ij}(t) = \pi_{ij} := \begin{cases} \frac{dL_{ij}(0)}{\sum_{k \in \ncal_0} dL_{ik}(0)} &\text{if } i \in \ncal \\ \frac{1}{n}\ind{j \in \ncal} &\text{if } i = 0\end{cases}\]
for all times $t$ and banks $i,j \in \ncal_0$. 
To further simplify the dynamics of this system, we will set $\dcal_i \equiv 1$ for every bank $i$.  That is, we consider the network dynamics introduced in Section~\ref{sec:continuous} alone without the need to consider bank defaults.
This setting (independently) replicates the results from \cite{sonin2020continuous} on a differential formulation of the static Eisenberg-Noe system.

Further, expanding and solving the differential system \eqref{eq:continuous-V}, we deduce that the continuous-time clearing cash accounts must satisfy the fixed point problem
\begin{equation}\label{eq:continuous-static}
V(t) = V(0) + x(t) + L(t)^\T\vec{1} - L(t)\vec{1} - \Pi^\T V(t)^-
\end{equation}
at all time $t \in \bbt$ with the additional assumption that $x_i(0) = L_{ij}(0) = 0$ for $i \in \ncal$ and $j \in \ncal_0$.  Therefore, if $x(t) \geq \vec{0}$ at some time $t$, it follows that $V(t)$ are the \emph{static} clearing accounts to the Eisenberg-Noe system with aggregated data with nominal liabilities matrix defined by $L(t)$ and external assets given by $x(t)$.  Importantly, if the relative liabilities are kept constant over time, taking aggregated data and considering the static Eisenberg-Noe framework will produce the same \emph{final} clearing cash accounts as the dynamic Eisenberg-Noe setting presented in this section.  However, though the set of delinquent banks at the terminal time is the same the set of defaulting banks in the static setting, the order of delinquencies need not strictly follow the order given in the fictitious default algorithm of \cite{EN01}.

\begin{definition}\label{defn:order-default}
A bank is called a \textbf{\emph{$k$th-order default}} in the static Eisenberg-Noe setting if it is determined to be in default in the $k$th iteration of the fictitious default algorithm (see, e.g., \cite[Section 3.1]{EN01} and included in Appendix~\ref{sec:FDA}).
\end{definition}
We note that the first-order defaults are exactly those firms that have negative cash account even if it has no negative exposure to other firms (i.e., all other firms satisfy their obligations in full).

\begin{proposition}\label{prop:order-defaults}
Let $(x,\bar{L}) \in \bbr_+^{n+1} \times \bbr^{(n+1) \times (n+1)}_+$ denote the static external assets and nominal liabilities.  Define a dynamic system over the time period $\bbt = [0,T]$ such that $V(0) \in [0,x]$, $dL(t) = \frac{1}{T}\bar{L}dt$, and $dx(t) = \frac{1}{T}\left(x - V(0)\right)dt$.  The clearing cash accounts at the terminal time $V(T)$ are equal to those given in the static setting.  Additionally, no firm will ever recover from delinquency in the dynamic setting.  Finally, the first $k$th-order default will become delinquent only after the first $(k-1)$th-order default in the static fictitious default algorithm; in particular, the first firm to become delinquent will be a first-order default in the static fictitious default algorithm.
\end{proposition}
\begin{proof}
The fact that the clearing cash accounts $V(T)$ are equal to the static Eisenberg-Noe clearing cash accounts (as defined in Proposition~\ref{prop:EN-e}) follows from \eqref{eq:continuous-static} (cf.\ \eqref{eq:EN-e} for the static system). Additionally, since $dx(t)$ is constant in time and firms are beginning in a solvent state, over time the unpaid liabilities may accumulate as a negative factor on bank balance sheets, but there is no outlet to allow for a firm to recover from delinquency.  Finally, by definition, a $k$th-order default is only driven into delinquency through the failure of the $(k-1)$th-order defaults (and \emph{not} solely by the $(k-2)$th-order defaults).  Therefore, by way of contradiction, if a $k$th-order default were to occur before any $(k-1)$th-order default then such a firm must default without regard to what happens to the $(k-1)$th-order defaults, i.e., this firm must be a $(k-1)$th-order default.  By this same logic, the first firm to become delinquent must be a first-order default.
\end{proof}

That the order in which banks enter delinquency differs from the order introduced by the fictitious default algorithm is unsurprising.  Consider a financial system with two subgraphs that are only connected through their obligations to the societal node.  By construction, the default or delinquency of a firm in one subgraph will have no impact on the firms in the other subgraph.  Thus we can construct a network so that all delinquencies in one subgraph (including higher order defaults as defined in Definition~\ref{defn:order-default}) occur before any first-order defaults in the other subgraph.

Notably, the proof of Proposition~\ref{prop:order-defaults} states that, provided the aggregate data (until the terminal time) is kept constant, the clearing cash accounts at the terminal time will be path-independent in this setting.  We will demonstrate this with an illustrative example demonstrating this setting in a small 4 bank (plus societal node) system.  In particular, we will consider the assets $x$ to be defined as (the discrete sampling of) a Brownian bridge so as to provide the appropriate aggregate data at the terminal time.
\begin{example}\label{ex:differentialEN}
Consider a financial system with four banks, each with an additional obligation to an external societal node.  Consider the time interval $\bbt = [0,1]$ with aggregated data such that the initial cash accounts are given by $V(0) = (100,1,3,2,5)^\T$, cash flows $dx$ are such that $x(0) = x(1) = \vec{0}$, and where the nominal liabilities matrix $dL = \bar{L}dt$ is defined by
\[\bar{L} = \left(\begin{array}{ccccc}0 & 0 & 0 & 0 & 0 \\ 3 & 0 & 7 & 1 & 1 \\ 3 & 3 & 0 & 3 & 3 \\ 3 & 1 & 1 & 0 & 1 \\ 3 & 1 & 2 & 1 & 0 \end{array}\right).\]
The \emph{static} Eisenberg-Noe clearing cash accounts, with nominal liabilities $\bar{L}$ and external assets $V(0)$, are found to be $V(1) \approx (109.38,-6.81,-3.03,-0.32,1.62)^\T$.  Further, from the static fictitious default algorithm, we can determine that bank 1 is a first-order default, bank 2 is a second-order default, and bank 3 is a third-order default.  Consider now three dynamic settings which are differentiated only by the choice of the cash flows $dx$:
\begin{enumerate}
\item Consider the deterministic setting introduced in Proposition~\ref{prop:order-defaults}, i.e., $dx(t) = \vec{0}$ for all times $t \in \bbt$.
\item Consider a Brownian bridge with low volatility, i.e., $dx(t) = -\frac{x(t)}{1-t}dt + dW(t)$ for vector of independent Brownian motions $W$.
\item Consider a Brownian bridge with high volatility, i.e., $dx(t) = -\frac{x(t)}{1-t}dt + 5dW(t)$ for vector of independent Brownian motions $W$.
\end{enumerate}
\begin{figure}[t]
\centering
\begin{subfigure}[t]{0.3\linewidth}
\centering
\includegraphics[width=\linewidth]{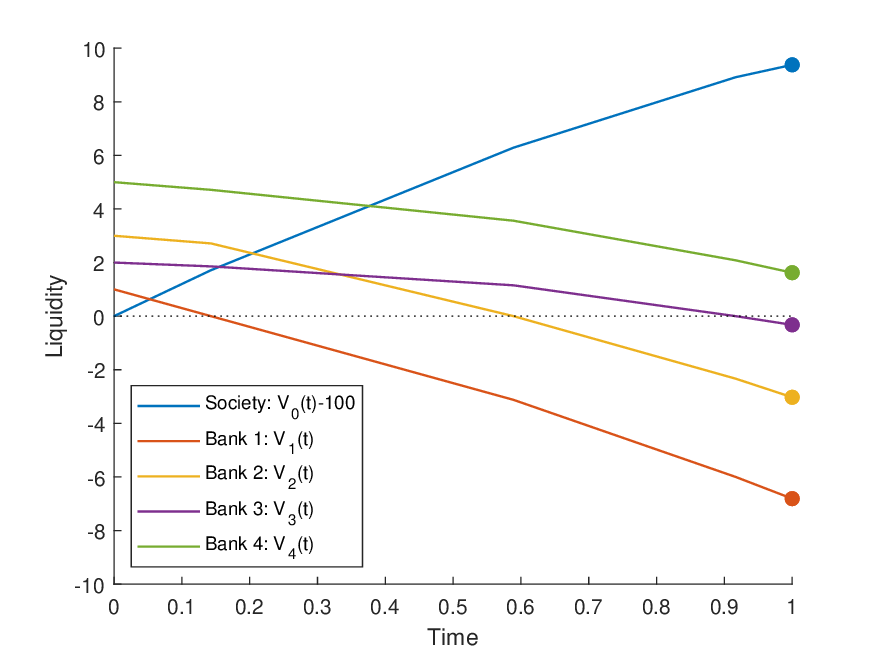}
\caption{Example~\ref{ex:differentialEN}: Clearing cash accounts over time under deterministic and constant external assets.}
\label{fig:EN-BB0}
\end{subfigure}
~
\begin{subfigure}[t]{0.3\linewidth}
\centering
\includegraphics[width=\linewidth]{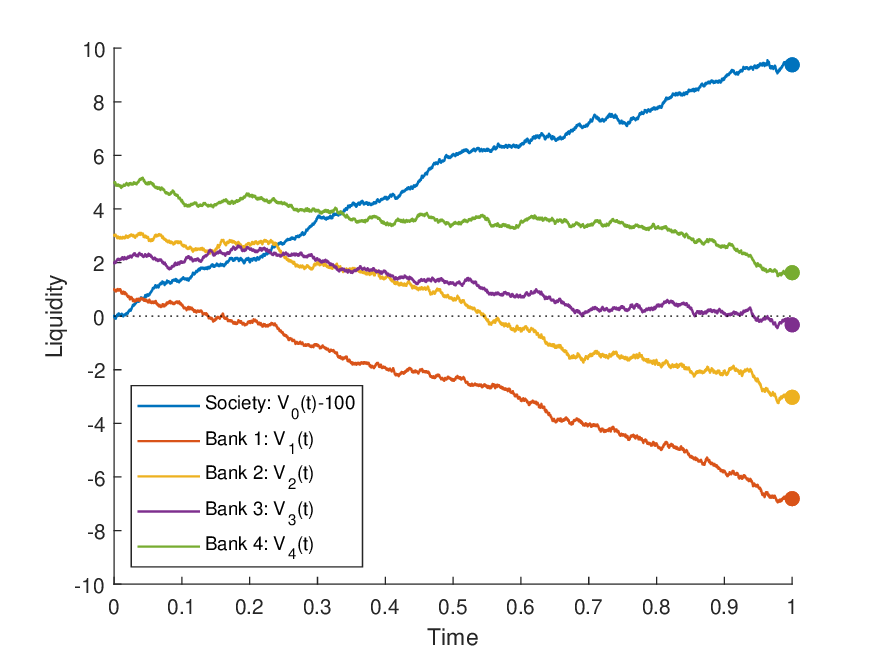}
\caption{Example~\ref{ex:differentialEN}: Clearing cash accounts over time under low volatility Brownian bridge external assets.}
\label{fig:EN-BB1}
\end{subfigure}
~
\begin{subfigure}[t]{0.3\linewidth}
\centering
\includegraphics[width=\linewidth]{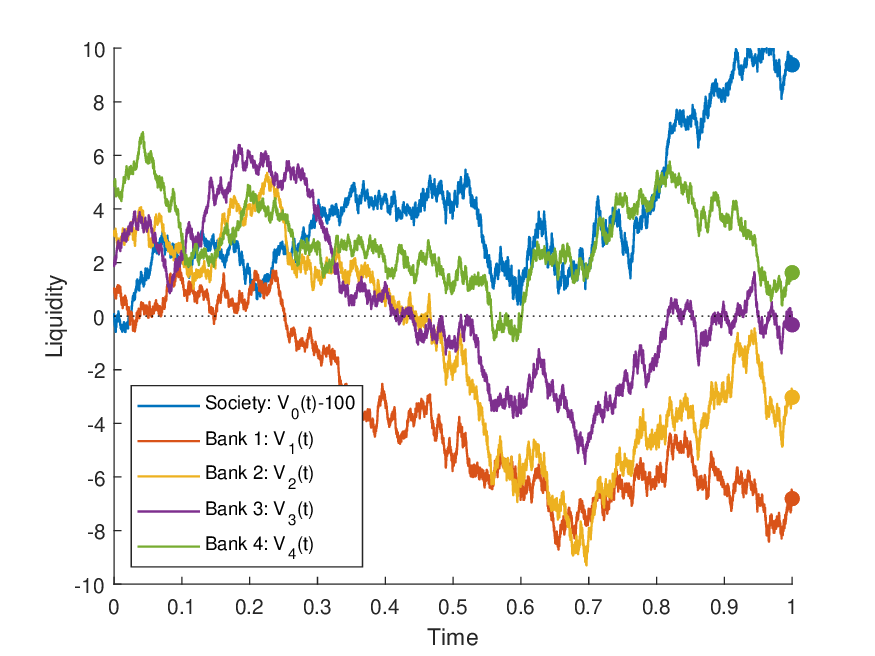}
\caption{Example~\ref{ex:differentialEN}: Clearing cash accounts over time under high volatility Brownian bridge external assets.}
\label{fig:EN-BB5}
\end{subfigure}
\caption{Example~\ref{ex:differentialEN}: Comparison of clearing cash accounts under deterministic and random cash flows that aggregate to the same terminal values.}
\label{fig:EN}
\end{figure}
A single sample path for each dynamic setting is provided.  In each plot we reduce the equity of the societal node by 100 so that it begins with an initial cash account of 0, but more importantly so that it can easily be displayed on the same plot as the other 4 institutions.  First, we point out that, as indicated by the circles at the terminal time in each plot, the terminal cash accounts of the continuous-time setting match up with the clearing cash accounts in the static model.   We further note that in the deterministic setting (Figure~\ref{fig:EN-BB0}) and the low volatility setting (Figure~\ref{fig:EN-BB1}) the order of delinquencies is maintained.  However, in the high volatility setting (Figure~\ref{fig:EN-BB5}) the order of delinquencies given by the fictitious default algorithm no longer holds.
\end{example}

\section{Static Fictitious Default Algorithm}\label{sec:FDA}
The static Eisenberg-Noe system can be computed via the, so-called, \emph{fictitious default algorithm}.  Consider the setup presented in Section~\ref{sec:setting}.  This result is presented for completeness.
\begin{alg}\label{alg:FDA}
Under the setting of Section~\ref{sec:setting}, the clearing cash accounts can be found by the following algorithm.  Initialize $k = 0$, $V^0 = x + L^\T \vec{1} - L \vec{1}$ and $D^0 = \emptyset$.  Repeat until convergence:
\begin{enumerate}
\item Increment $k = k+1$.
\item Denote the set of insolvent banks by $D^k := \{i \in \ncal \; | \; V_i^{k-1} < 0\}$.
\item If $D^k = D^{k-1}$ then terminate and set $V = V^{k-1}$.
\item Define the matrix $\Lambda^k \in \{0,1\}^{n \times n}$ so that $\Lambda_{ij}^k = \begin{cases}1 &\text{if } i = j \in D^k \\ 0 &\text{else} \end{cases}$.
\item Define $V^k = (I - \Pi^\T \Lambda^k)^{-1}(x + L^\T \vec{1} - L \vec{1})$.
\end{enumerate}
\end{alg}
A bank $i$ is defined as being a $k$th-order default if $i \in D^k \backslash D^{k-1}$.

\end{document}